\documentclass[12pt]{article}

\usepackage[dvipsnames]{xcolor}

\usepackage{amsmath, amsthm, amssymb}
\usepackage{fullpage}
\usepackage{graphics}
\usepackage{graphicx}
\usepackage{rotating}
\usepackage{mathrsfs}
\usepackage{enumerate}
\usepackage{natbib}
\usepackage{color,soul}

\usepackage{enumitem}
\setlist[enumerate]{label = \textbf{\alph*)}}

\newcommand{\Ocal}{\mathcal{O}}

\def\P{P}
\def\R{\mathbb R}
\def\Q{\mathbb Q}

\def\E{E}

\def\Var{\text{Var}}
\def\Cov{\text{Cov}}

\def\asto{\stackrel{a.s.}{\longrightarrow}}

\def\bfalpha{\boldsymbol{\alpha}}
\def\bfbeta{\boldsymbol{\beta}}
\def\bfth{\boldsymbol{\theta}}
\def\blam{{\boldsymbol{\lambda}}}
\def\bfa{\mathbf{a}}
\def\bfb{\mathbf{b}}
\def\bfh{\mathbf{h}}
\def\hbaseline{\eta}
\def\Tghost{\widetilde{T}}
\def\etaA{A}
\def\vecA{\mathbf{ Z_1 }}
\def\vecZ{\mathbf{ Z_2 }}

\def\OR0{e^{\bfalpha_0^\top\vecA_i}}
\def\phazi{e^{\bfbeta^\top\vecZ_i}}
\def\phaz0{e^{\bfbeta_0^\top\vecZ_i}}
\def\phazE{e^{\bfbeta_0^\top\vecZ}}
\def\bfsig{\boldsymbol{\sigma}}
\def\lowm{m^{-1}}
\def\upM{m}

\newcommand{\be}{\begin{eqnarray}}
\newcommand{\ee}{\end{eqnarray}}
\newcommand{\bes}{\begin{eqnarray*}}
\newcommand{\ees}{\end{eqnarray*}}

\ifx\BlackBox\undefined
\newcommand{\BlackBox}{\rule{1.5ex}{1.5ex}}  
\fi

\ifx\QED\undefined
\def\QED{~\rule[-1pt]{5pt}{5pt}\par\medskip}
\fi

\ifx\proof\undefined
\newenvironment{proof}{\par\noindent{\bf Proof\ }}{\hfill\BlackBox\\[2mm]}
\fi

\ifx\theorem\undefined
\newtheorem{theorem}{Theorem}
\fi
\ifx\condition\undefined
\newtheorem{condition}{Condition}
\fi
\ifx\definition\undefined
\newtheorem{definition}{Definition}
\fi
\ifx\example\undefined
\newtheorem{example}{Example}
\fi
\ifx\property\undefined
\newtheorem{property}{Property}
\fi
\ifx\lemma\undefined
\newtheorem{lemma}{Lemma}
\fi
\ifx\proposition\undefined
\newtheorem{proposition}[theorem]{Proposition}
\fi
\ifx\remark\undefined
\newtheorem{remark}{Remark}
\fi
\ifx\corollary\undefined
\newtheorem{corollary}[theorem]{Corollary}
\fi
 \ifx\conjecture\undefined
\newtheorem{conjecture}[theorem]{Conjecture}
\fi
\ifx\fact\undefined
\newtheorem{fact}[theorem]{Fact}
\fi
\ifx\claim\undefined
\newtheorem{claim}[theorem]{Claim}
\fi

\newtheorem{assumption}{Assumption}

\newtheorem{innercustomthm}{Theorem}
\newenvironment{customthm}[1]
  {\renewcommand\theinnercustomthm{#1}\innercustomthm}
  {\endinnercustomthm}

\newtheorem{innercustomass}{Assumption}
\newenvironment{customass}[1]
  {\renewcommand\theinnercustomass{#1}\innercustomass}
  {\endinnercustomass}

\begin{document}

\begin{center}
\large \textbf{A Nonparametric Maximum Likelihood Approach for Partially Observed Cured Data with Left Truncation and Right-Censoring}\normalsize\\~\\
\textsc{ Jue Hou\textsuperscript{1}, Christina D.~Chambers\textsuperscript{2,3} and Ronghui Xu\textsuperscript{1,2*} }\\
\textsuperscript{1}\textit{Department of Mathematics, }\\
\textsuperscript{2}\textit{Department of Family Medicine and Public Health, }\\
\textsuperscript{3}\textit{Department of Pediatrics, University of California, San Diego}\\
\textsuperscript{*}\textit{Corresponding author: rxu@ucsd.edu }
\end{center}
\date{today}

\abstract

Partially observed cured data occur in the analysis of spontaneous abortion (SAB) in observational studies in pregnancy. In contrast to   the traditional cured data, such data has an observable `cured' portion as women who do not abort spontaneously.  The data is also subject to left truncate in addition to right-censoring because women may enter or withdraw from a study any time during their pregnancy. Left truncation in particular causes unique bias in the presence of a cured portion.
In this paper, we study a cure rate model and develop a conditional nonparametric maximum likelihood approach.  To tackle the computational challenge we adopt an EM  algorithm making use of  ``ghost copies" of the data, and a closed form variance estimator  is derived.  Under suitable assumptions, we prove the consistency of the resulting estimator involving an unbounded cumulative baseline hazard function, as well as the asymptotic normality. Simulation results are carried out to evaluate the finite sample performance.  We present the analysis of the motivating SAB study to illustrate the power of our model addressing both occurrence and timing of SAB, as compared to existing approaches in practice.


\section{Introduction}
Our work was motivated by research  carried out at the Organization of Teratology
Information Specialists (OTIS), which is a North American network of university or hospital
based teratology services that counsel between 70,000 and 100,000 pregnant women every
year.
Research subjects are enrolled from the Teratology Information Services and through
other methods of recruitment, where the mothers and their babies are followed over time.
Phone interviews are conducted through the length of the pregnancy along with pregnancy
diaries recorded by the mother.
An outcome phone interview is conducted shortly after the pregnancy ends, and if it results in a live birth, a dysmorphology exam is done within six
months and with further follow-ups at one year and possibly later dates. Recently it has
been of interest to assess the effects of medication exposures on spontaneous abortion (SAB)
\cite[]{Xu11, Chambers11}. Here we examine the OTIS autoimmune disease in pregnancy database for risk factors as well as
effects of medications  on spontaneous abortion.

By definition SAB occurs within the first 20 weeks of gestation; any spontaneous pregnancy loss after
that is called still birth. Ultimately we would like to know if an exposure modifies the
risk of SAB for a woman, which may be increased or decreased. It is known that in the
population for clinically recognized pregnancies the rate of SAB is about 12\% \cite[]{Wilcox88}.
On the other hand, in our database the empirical SAB rate is consistently lower
than 10\%. This is due to the fact that women may enter a study any time before 20 weeks'
gestation. Figure \ref{fig:Qhist} left panel shows the histograms of study entry times up to 20 weeks of gestation
from our autoimmune disease in pregnancy database. This way women who have early SAB events are
less likely to be captured in our studies, and such selection bias is known as left truncation
in survival analysis. Left truncation has been studied by many authors since the 1980s, and
has attracted much recent attention in the context of length-biased data \cite[among others]{Asgharian06, Qin11}.
Figure \ref{fig:Qhist} right panel shows the left truncated Kaplan-Meier curve
for the SAB event.

As seen from the Kaplan-Meier curve
the majority of the pregnant women are free of SAB; they are  considered `cured'
in the time-to-event context. Cure rate models are well studied in the literature for right-censored
data. The models effectively analyze the survival distribution of
those who are susceptible along with the probability of an individual being `cured'. In the
approaches using mixture models,  logistic regression is often used to model the cured probability.
For the dependency of the survival function on the covariates among the non-cured,
various regression models have been considered: the Cox proportional hazards model
\cite[]{Kuk92,Sy2000}, transformation models \cite[]{Lu04},
and richly parametrized models when the shape of the hazard function is of interest \cite[]{Hanson03}.
Cure rate models have also been developed along the lines of non-mixture
models \cite[]{Chen99, Zeng06}.
In addition to right-censored data, cure-rate models have also been developed for interval-censored
data \cite[]{Kim08}. To our best knowledge, however, they have not been
considered for truncated data which, unlike censoring, poses a unique set of challenges. While
left truncation has been well studied in the literature, the challenges are again unique in
the presence of a cured portion. Most importantly, left truncation leads to selection bias that
needs to be explicitly counted for, and in the process of doing so computational challenges
also arise, as will be seen below.

Cure models are used in various biomedical studies where data often include a
substantial portion of `long-term' survivors who are no longer susceptible to the event of
interest \cite[]{Farewell82, Farewell86}. Our data however, differs from classical cured data where the `long-term' survivors are never observed to be cured, rather they are censored at a finite time point \cite[often taken as the maximum]{Sy2000, Lu04}.
In our case, `cured' is defined as surviving 20 weeks of gestation, and we observe over 80\% of our subjects as cured from SAB.

In the following we consider the mixture cure rate model. This choice has been made
based on in-depth discussions with our scientific collaborators, because it is important to
understand both the risk factors for SAB (yes/no) as well as the predictors of timing of SAB
events among those who experience them. Different timing of SAB can reflect different
underlying biological processes.
In the next section 
we show that with many observed `cured' women in our
data, a slightly different likelihood than the one commonly seen in the literature should be
used. We discuss computational challenges with the likelihood, and adopt an EM algorithm using `ghost copies' of the observed data.
In section \ref{section:theory}, the resulting estimator is shown to be consistent and asymptotic normal, despite the fact that the cumulative baseline hazard function diverges at the finite time point before `cure' is achieved.
 We illustrate the effectiveness of the method on finite samples via simulation experiments in section \ref{section:simulation}. 
 We  conclude  with  the  analysis  of  SAB  data  from  the  OTIS  database described above.

\section{Model and NPMLE}\label{section:model}

    \subsection{Model and partially observed cured data }

    Let $\tau<\infty$ be a strict upper bound of time for the event of interest, beyond which a subject is considered cured. In the pregnancy example above, this would be the 20 weeks of gestation. The whole population consists of two subpopulations: cured and non-cured.
    Let the binary random variable $\etaA$ indicate whether a subject belongs to the non-cured subpopulation;
    and let 
 $T^*\in (0,\tau)$ be the failure time random variable for this subpopulation.
    The overall outcome time $T$ is given by the mixture 
    \cite[]{Lu04}:
   $ T=\etaA T^*+(1-\etaA)\tau$.
    Let $\vecA$ and $\vecZ$
    be 
       two covariate vectors;
    they may share common covariates depending on the application.
    We assume that $\etaA$ given $\vecA$ follows the logistic regression model
    \begin{equation*}
    \P(\etaA=1|\vecA,\vecZ)=p=\frac{e^{\bfalpha^\top\vecA}}{1+e^{\bfalpha^\top\vecA}},
    \end{equation*}
    and that $T^*$ given $\vecZ$ follows the proportional hazard regression model
    with cumulative baseline  hazard function $\lambda_0(t)=\int_0^t \lambda_0(u)du$:
    \begin{equation*}
    \P(T^* \ge t | \vecZ)=S(t | \vecZ)=\exp\{-\Lambda_0(t)e^{\bfbeta^\top\vecZ} \}.
    \end{equation*}
    Note that $ \Lambda_0(\tau)=+\infty $ so that $ S(\tau | \vecZ) = 0$.

Our data is subject to left truncation  and right-censoring.
    Let $Q$ be the left truncation time and $C$ the right-censoring time,
    satisfying $0 \le Q < C $;
     we also assume that they
    are independent of $( \etaA, T^*)$ conditioning on $\vecA$ and $\vecZ$.
    For subjecrts $i=1, ..., n$, the observed data include $\vecA_i$, $\vecZ_i$, $Q_i$, $X_i=T_i\wedge C_i$, $\delta^1_i=\etaA_i \cdot I(T_i \le C_i)$,
        $\delta^0_i=(1-\etaA_i)I(C_i\geq \tau)$ and $\delta^c_i=I(C_i<T_i \leq \tau)$.
        In other words $\delta^1_i$ is the indicator that a subject has an observed event (non-cured),
        $\delta^0_i$ is the indicator that a subject is observed to be cured,  and $\delta^c_i$ is the indicator that a subject is censored before $\tau$ so that we do not know whether she is cured or not.  
        Note that the subject $i$ is observed only if $T_i>Q_i$, hence left truncation is known to lead to a biased sample from the population.
        We  note again that our data is different from the classical cure model literature, where the cured individuals are always treated as censored; we refer to our data as {\it partially observed cured data}. Because of right-censoring, $A_i$ may not be observed; but we emphasize here that we do observe many $A_i=0$ in our data.

    Denote $\bfth=(\bfalpha,\bfbeta,\Lambda_0)$.   For the purposes of nonparametric maximum likelihood estimation (NPMLE), it is necessary to  discretize $\Lambda_0 $  to be $ \Lambda_0(t)=\sum_{k=1}^K \lambda_k I(t\ge t_k)$, where
         $0< t_1< \dots <t_K <\infty$ are the unique failure times
	 \cite[]{Johansen, Murphy94}.
    We apply the likelihood approach conditional upon the left truncation time $Q_i$ and the right-censoring time $C_i$, as no parametric distributional assumptions are made about these two random variables.
    Denote  $p_i=e^{\bfalpha^\top\vecA_i}/(1+e^{\bfalpha^\top\vecA_i})$,
       $\lambda_i(t)=\lambda_0(t)\exp(\bfbeta^\top \vecZ_i)$,
     $f_i(t)=\lambda_i(t) S_i(t)$, 
      and  $S_i(t)= \exp \{-\Lambda_0(t)\phazi \}$.
    The likelihood for our observed data is 
    \be \label{Lik:obs}
      L(\bfth) &=& \prod_{i=1}^n L_i(\bfth; X_i, \delta^1_i, \delta^0_i, \delta^c_i | T_i>Q_i, \vecA_i, \vecZ_i, Q_i, C_i) \nonumber\\
      &=& \prod_{i=1}^n \frac{\big\{p_i\lambda_i(X_i)S_i(X_i)\big\}^{\delta^1_i}
        (1-p_i)^{\delta^0_i}\big\{1-p_i+p_iS_i(X_i)\big\}^{\delta^c_i} }{1-p_i+p_iS_i(Q_i)}, 
    \ee
where $ 1-p_i+p_iS_i(X_i) = P(T_i>Q_i)$. 

    \subsection{NPMLE through EM}
    \subsubsection*{\it \underline{Complete data likelihood}}

    The complexity of observed likelihood \eqref{Lik:obs} leads to the challenge of optimization.
    To reduce the  problem 
        we follow the approach of \cite{Vardi85}, rediscovered recently by \cite{Qin11}.

   To augment the observed data, we first note that the group indicator $\etaA_i$
     is latent whenever censoring occurs.
    In addition, we compensate for the left truncation  through the ``ghost copy" algorithm proposed in \cite{Qin11}.
    For each observed subject with the pair of covariates $(\vecA_i,\vecZ_i)$ and entry time $Q_i$,
    there are $M_i$ hypothetical ``truncated samples" with latent event time $\Tghost_{ij}<Q_i$, $ j= 1, ..., M_i$.
    The resulting complete likelihood is
    \begin{align}\label{Lik:complete}
     \notag    L^c(\bfth)&=\prod_{i=1}^n\big\{p_i\lambda_i(X_i)S_i(X_i)\big\}^{\delta^1_i}
        (1-p_i)^{\delta^0_i}\big\{p_iS_i(X)\big\}^{\etaA_i\delta^c_i} (1-p_i)^{(1-\etaA_i)\delta^c_i}\\
        &\times p_i^{M_i}\prod_{j=1}^{M_i}\prod_{k:t_k
               \le Q_i}\big\{\lambda_k e^{\bfbeta^\top\vecZ_i}S_i(t_k) \big\}^{I(\Tghost_{ij}=t_k)}
    \end{align}
    In this way, the two sets of parameters $\bfalpha$ and $\bfbeta$ are separated in the complete data likelihood.
    All remaining product terms are those in the usual likelihoods for the logistic and the Cox regression model.
    Consequently, the M-step update is instantly available from existing solvers.

 Given the observed data  $\Ocal$,  it can be seen that for  subject $i$ who is censored at $X_i$, the unobserved group indicator
    $\etaA_i$ follows Bernoulli distribution with $P(\etaA_i =1) = {p_iS_i(X_i)} / \{1-p_i+p_iS_i(X_i)\}$.
    For a subject with truncation time $Q_i$ and covariates $(\vecA_i,\vecZ_i)$,
   it can be seen that the number of truncated ``ghost'' copies $M_i$ follows the geometric distribution with probability  $ P(T_i < Q_i) = p_i\{1-S_i(Q_i)\}$.
    For the ``ghost'' event times 
     let $\Tghost_{ij}$ be one of the observed event times $t_k < Q_i$ with probability proportional to
     $f_i(t_k)=\lambda_k\phazi S_i(t_k)$:
    \begin{equation}\label{dist:k}
      \P(\Tghost_{ij}=t_k|M_i, \Ocal) =\frac{I(t_k\le Q_i)\lambda_k\phazi S_i(t_k)}
            {\sum_{k:t_k\le Q}\lambda_k\phazi S_i(t_k)}.
    \end{equation}
By restricting the ``ghost'' event times to the observed event times, we are able to exploit the convenience of directly applying the weighted Cox regression later. The price we pay is a slight discrepancy between  $\sum_{k:t_k\le Q}\lambda_k\phazi S_i(t_k)$ and $1-S_i(Q_i)$. Integrating out the latent variables in
    $L^c(\bfth)$ does not give exactly the observed likelihood $L(\bfth)$.
However,  we show later  that this difference is asymptotically negligible     so that the solution from the above EM is asymptotically equivalent to the true NPMLE.

    \subsubsection*{\it \underline{The EM Algorithm}}

    From \eqref{Lik:complete} we can write the complete data log-likelihood $l^c=\log L^c$ as
\begin{align}\label{loglik:EM}
\notag l^c(\bfalpha,\bfbeta,\blam)=& \sum_{i=1}^n \bigg[\delta^1_i \etaA_i \sum_{k=1}^K I\{X_i=t_k\}\log f_i(t_k) +M_i\sum_{k:t_k<Q_i} I\{\Tghost_i=t_k\}\log f_i(t_k) \\
 &  +(1-\delta^1_i)\etaA_i \log S_i(X_i)+(1-\etaA_i)\log(1-p_i)+ (\etaA_i+M_i) \log(p_i)\bigg],
\end{align}
where $ \blam = (\lambda_1, ..., \lambda_K)$.  

Though the algorithm runs stably from any initial values of the parameters in the support,
we recommend to fit a na\"{i}ve logistic regression without censored subjects for $\bfalpha^{(0)}$
and a na\"{i}ve Cox regression for $\bfbeta^{(0)}$ and $\blam^{(0)}$ treating the observed cured subjects as censored at $\tau$,
to minimize the number of iterations until convergence.

\textbf{E-step}

At the $(l+1)$-th iteration ($l=0, 1, ...$), let ${\bfalpha}^{(l)}, {\bfbeta}^{(l)}, {\blam}^{(l)}$ be the parameter values at the current iteration
    upon which ${p}_i^{(l)}$, ${f}_i^{(l)}$ and ${S}_i^{(l)}$ are defined.
The distributions of the latent variables conditioning on the observed data are given in the above, and
their conditional expectations can be computed as
\begin{align}
\E&[I\{\Tghost_{ij}=t_k\}| M_i, \mathcal{O}; {\bfalpha}^{(l)}, {\bfbeta}^{(l)}, {\blam}^{(l)} ]
    =\frac{I(t_k < Q_i) {f}^{(l)}_i(t_k)}{\sum_{h:t_h<Q_i} {f}^{(l)}_i(t_h)} \label{EM:condE1},\\
\E&[M_i| \mathcal{O}; {\bfalpha}^{(l)}, {\bfbeta}^{(l)}, {\blam}^{(l)} ]=
    \frac{ {p}^{(l)}_i\sum_{k:t_k<Q_i} {f}^{(l)}_i(t_k)}  {1- {p}^{(l)}_i\sum_{k:t_k<Q_i} {f}^{(l)}_i(t_k)},\\
\E&[\etaA_i|\mathcal{O}; {\bfalpha}^{(l)}, {\bfbeta}^{(l)}, {\blam}^{(l)} ]=
    \delta^1_i+\delta^c_i\frac{{p}^{(l)}_i {S}^{(l)}_i(X_i)}  {1- {p}^{(l)}_i+ {p}^{(l)}_i {S}^{(l)}_i(X_i)}\label{EM:condE3}.
\end{align}
Since the latent variables all enter linearly into the complete data log-likelihood, the expected complete data log-likelihood is
\begin{equation}\label{eq:expected_l}
E (l^c | \mathcal{O}) =\sum_{i=1}^n\sum_{k=1}^K \left\{w^f_{i,k}\log f_i(t_k) +w^S_{i} \log S_i(X_i)+w^p_{0,i}\log(1-p_i)+ w^p_{1,i} \log(p_i)\right\},
\end{equation}
where the weights are computed as
\begin{eqnarray*}
w^f_{i,k}&=& \delta^1_i I\{X_i=t_k\} + \frac{ {p}^{(l)}_i {f}^{(l)}_i(t_k)}  {1- {p}^{(l)}_i \sum_{h:t_h<Q_i} {f}^{(l)}_i(t_h)} I\{t_k<Q_i\},\\
w^S_i&=&\delta^c_i\frac{ {p}^{(l)}_i {S}^{(l)}_i(X_i)}  {1- {p}^{(l)}_i + {p}^{(l)}_i {S}^{(l)}_i(X_i)},\\
w^p_{0,i}&=& \delta^0_i + \delta^c_i \frac{1- {p}^{(l)}_i}  {1- {p}^{(l)}_i+ {p}^{(l)}_i {S}^{(l)}_i(X_i)},\\
w^p_{1,i}&=& \delta^1_i \etaA_i + \delta^c_i \frac{ {p}^{(l)}_i {S}^{(l)}_i(X_i)}  {1- {p}^{(l)}_i+ {p}^{(l)}_i {S}^{(l)}_i(X_i)}+
\frac{ {p}^{(l)}_i \sum_{k:t_k<Q_i} {f}^{(l)}_i(t_k)}  {1- {p}^{(l)}_i\sum_{k:t_k<Q_i} {f}^{(l)}_i(t_k)}.
\end{eqnarray*}

\textbf{M-step}

From \eqref{eq:expected_l} the expected log-likelihood can be written as the sum of two parts, so that the M-step can be achieved using
a weighted logistic regression optimized over $\bfalpha$:
\begin{equation*}
l_{glm}=\sum_{i=1}^n w^p_{0,i}\log(1-p_i)+ w^p_{1,i} \log(p_i);
\end{equation*}
and a weighted Cox proportional hazard regression optimized over
$\bfbeta$: 
\begin{equation*}
l_{coxph}=\sum_{i=1}^n \sum_{k=1}^K w^f_{i,k}\log f_i(t_k) +\sum_{i=1}^n w^S_{i} \log S_i(X_i).
\end{equation*}
Easily implemented solution is available from existing \emph{glm} and \emph{coxph} solvers in R,  to obtain $\bfalpha^{(l+1)}$, $\bfbeta^{(l+1)}$
and $\blam^{(l+1)}$.

\subsubsection*{\it \underline{Variance Estimator}}

At convergence of the EM algorithm where $ \hat{\bfth} $ denotes the NPMLE,  the
\cite{Louis82} formula can be used to give the observed Fisher information:
\begin{equation}\label{eq:louis}
I_{obs}(\hat{\bfth}) =
\sum_{i=1}^n \E_{\hat{\bfth}} [B_i|\mathcal{O}]
-\sum_{i=1}^n \E_{\hat{\bfth}}[\mathbf{S}_i \mathbf{S}_i^\top|\mathcal{O}]
-2\sum_{i<i'}^n \E_{\hat{\bfth}}[\mathbf{S}_i |\mathcal{O}]
\E_{\hat{\bfth}}[\mathbf{S}_{i} |\mathcal{O}]^\top,
\end{equation}
where $\mathbf{S}_i$ and $B_i$ are  the gradient $\nabla l^c_i$
    and  the negatives of Hessian $-\nabla^2 l^c_i$ of the complete data log-likelihood.
The above is  in closed form, and the
details are given in Appendix \ref{appendixB}. We show in the next section that \eqref{eq:louis} provides a consistent variance estimator for the NPMLE, and its use in association with the NPMLE has been advocated in the literature  \cite[]{vaid:xu:00, zeng:lin:07, gamst:etal}.

\section{Theory}\label{section:theory}

Let  $\bfth_0=(\bfalpha_0,\bfbeta_0,\Lambda_0(\cdot))$ denote the true parameter value.
 Following \cite{ande:borg:93},
    we define the counting process $N_i(t)=\delta^1_iI(X_i \le t)$ and
       the at-risk process $Y_i(t)=I(Q_i \le t \le X_i)$.
    Their sums are denoted as $\bar{N}(t) = \sum_{i=1}^n N_i(t)$,
        and $ \bar{Y}(t) = \sum_{i=1}^n Y_i(t)$.
    By Doob-Meyer decomposition, a martingale with respect to the filtration
        $\mathcal{F}_t=\sigma\{N_i(u),Y_i(u),\vecA,\vecZ, u \le t\}$ is
    \begin{equation}\label{def:mart1}
    M_i(t)=N_i(t)-\int_0^t
\phi_i^{\bfth_0}(u) Y_i(u)e^{\bfbeta_0^\top\vecZ_i} d\Lambda_0(u),
    \end{equation}
    where
    \begin{equation}\label{def:phi}
    \phi_i^{\bfth}(t)=\frac{ \exp\{ {\bfalpha^\top\vecA_i -\Lambda(t)e^{\bfbeta^\top\vecZ_i}} \} }{1+ \exp \{ {\bfalpha^\top\vecA_i -\Lambda(t)e^{\bfbeta^\top\vecZ_i}} \} }
=\P_{\bfth}(\etaA_i = 1 | X_i \ge t).
    \end{equation}
 To make use of the martingale framework,
        we write the observed log-likelihood $l_n=\log L$, where $ L(\bfth)$ was given in \eqref{Lik:obs}, as
    \begin{align*} 
    \notag l_n=&\sum_{i=1}^n
    \int_0^\tau \log\Big(\phi_i^{\bfth}(u)e^{\bfbeta^\top\vecZ_i}\Big)dN_i(u)
    -\int_0^\tau Y_i(u)\phi_i^{\bfth}(u)e^{\bfbeta^\top\vecZ_i}d\Lambda(u)\\
    &+\int_0^\tau \log\Big(\Delta \Lambda(u)\Big)dN_i(u),
    \end{align*}
where $\Delta \Lambda(u)$ is the size of jump of the baseline cumulative hazard  at $u$
\cite[]{Murphy94}.
    We establish the theory  under the following assumptions. 
The vector norm throughout this paper is the uniform norm, i.e. the largest absolute value among all elements. 

    \begin{assumption}\label{assumpt:parameter}
The true finite-dimensional parameter $(\bfalpha_0, \bfbeta_0)$ is an element of the interior of a 
compact set $\{(\bfalpha, \bfbeta) : \|\bfalpha\| \vee \|\bfbeta\| \le D_1\}$ for some constant $D_1$. 
    \end{assumption}

   \begin{assumption}\label{assumpt:cov bound}
      The covariates $(\vecA,\vecZ)$ follow distribution $F_{Z}(\cdot , \cdot)$.
            They are bounded a.s.:
      there  exists $D_2>0$, such that $\P(\max\{\|\boldsymbol{\vecA}\|,\|\boldsymbol{\vecZ}\| \} \le D_2)=1$.
      Also, their covariance matrices $\Var(\vecA)$(without intercept term) and $\Var(\vecZ)$ are both positive-definite.
    Denote constant $\upM$ such that 
   \begin{equation}\label{def:Mm}
    0<\lowm=e^{-D_1D_2}\le e^{\bfalpha^\top\vecA}\wedge e^{\bfbeta^\top\vecZ}
    \le e^{\bfalpha^\top\vecA}\vee e^{\bfbeta^\top\vecZ} \le e^{D_1D_2}=\upM < \infty \quad a.s..
    \end{equation}
    \end{assumption}

 \begin{assumption}\label{assumpt:baseline}
    The baseline cumulative hazard function $\Lambda_0(t)$ is a non-decreasing continuous function on $[0,\tau)$.
      $\Lambda_0(0)=0$ and $\Lambda_0(\tau-)=\infty$. And
      \begin{equation}\label{def:vareps}
        \inf_{t \in [0,\tau]}\E [Y(t)|\vecA,\vecZ]> \varepsilon >0, \quad a.s..
      \end{equation}
    \end{assumption}

\begin{assumption}\label{assumpt:ltrc}
 There exists     $ \zeta \in (0,\tau)$ such that $\P(Q>\zeta)=0$. 
$\Lambda_0(t)$ is strictly increasing over $[0,\zeta]$, and
      $\E[Y(t)|\vecA,\vecZ]$ is Lipschitz continuous w.r.t to $\Lambda_0(t)$ on $[0, \zeta]$ a.s.; that is, 
      there is a constant
      \begin{equation}\label{def:Lips}
        \mathcal{L}\ge \sup_{0\le t < s \le \zeta}\left\{\frac{|\E[Y(t)|\vecA,\vecZ]-\E[Y(s)|\vecA,\vecZ]|}
            {|\Lambda_0(t)-\Lambda_0(s)|}\right\}
            , \quad a.s..
      \end{equation}
    \end{assumption}

    The above  Assumption \ref{assumpt:baseline} is specifically made for cure rate models with an observable cured portion.
    This assumption enforces that the failure time must occur prior to a well-defined upper bound.
  Equation \eqref{def:vareps} requires that certain proportion of subjects enter the study at time zero.
While this may not always be the case for our pregnancy studies,  time zero may be replaced by the earliest entry time into the study and
the inference is conditional upon survival beyond that time, and all the results established in this section carry over.
Assumption \ref{assumpt:ltrc} gives the regularity conditions on truncation and censoring.
The truncation times should be bounded away from time $\tau$; this is required in order to establish Lemma \ref{lemma:EM} below.
The truncation-censoring distribution also has to possess certain level of continuity with respect to
the distribution of event time.
For example, the continuity condition is satisfied when the distributions for $Q$, $C$ and $T$ given $\vecA$ and $\vecZ$ all have densities that are bounded away 
from $\infty$ and $0$ almost surely. 
This condition can be weakened to allow $\Lambda_0(t)$ to be constant over some open set and require only that $\E[Y(t)|\vecA,\vecZ]$ is Lipschitz continuous with respect to $\Lambda_0(t)$ on a open set $\Omega \subset [0,\zeta]$ consisting of finite many open intervals, on which $\int_\Omega d\Lambda_0 = \Lambda_0(\zeta)$. 
All theoretical results under this weakened condition can be achieved by repeatedly applying the  steps in the current proof. 

    For the asymptotic normality we make the following assumption where $\tau'$ is defined later.
    \begin{customass}{\ref{assumpt:baseline}'}\label{assumpt:baseline'}
    The baseline cumulative hazard $\Lambda_0(t)$ is a non-decreasing continuous function on $[0,\tau']$.
      $\Lambda_0(0)=0$, $\Lambda_0(\tau') < \infty $  and $\Lambda_0(\tau-)=\infty$. And
      \begin{equation*}
        \inf_{t \in [0,\tau']}\E [Y(t)|\vecA,\vecZ]>\varepsilon>0, \quad a.s..
      \end{equation*}
    \end{customass}

\subsection{Existence of NPMLE} 

First, we show the existence of the NPMLE.
\begin{theorem}\label{thm:exist}
    Under Assumptions \ref{assumpt:parameter} and \ref{assumpt:cov bound},
    if $\sum_{i=1}^n N_i(\tau)>0$, then a maximizer of $l_n(\bfth)$, $\hat{\bfth}=(\hat{\bfalpha},\hat{\bfbeta},\hat{\Lambda}(\cdot))$ exists and is finite.
\end{theorem}
\noindent For the proof we use the same technique as in \cite{Murphy94}.
All the proofs are in Appendix \ref{appendixA}.

We now show that the solution from the previously described EM algorithm is  asymptotically equivalent to 
the NPMLE. 
\begin{lemma}\label{lemma:EM}
    Let $\tilde{\bfth}$ be the solution from the EM algorithm with complete data likelihood \eqref{Lik:complete}
    and $\hat{\bfth}$ be the NPMLE for the observed likelihood \eqref{Lik:obs}.
   Under Assumptions \ref{assumpt:parameter} 
   - \ref{assumpt:ltrc},
    $n^{-1} \{ l_n(\hat{\bfth})-l_n(\tilde{\bfth}) \} = O_p(1/n)$.
\end{lemma}
\begin{theorem}\label{cor:EMconsistency}
Under Assumptions \ref{assumpt:parameter}
   - \ref{assumpt:ltrc},
$\|\hat{\bfth}-\tilde{\bfth}\| = o_p(1)$.
\end{theorem}

\begin{customthm}{\ref{cor:EMconsistency}'}\label{cor:EMweakconv}
Under Assumptions \ref{assumpt:parameter}, \ref{assumpt:cov bound}, \ref{assumpt:ltrc}
    and \ref{assumpt:baseline'}, $\hat{\bfth}-\tilde{\bfth} = o_p(1/\sqrt{n})$.
\end{customthm}

\subsection{Consistency of NPMLE}

Next, we show the consistency of the NPMLE.
\begin{theorem}\label{thm:consistency}
Under Assumptions \ref{assumpt:parameter} - 
\ref{assumpt:ltrc},
    the NPMLE estimator for
    $L$ in \eqref{Lik:obs}, $\hat{\bfth}=(\hat{\bfalpha},\hat{\bfbeta},\hat{\Lambda}(\cdot))$,  is consistent.
    That is
    \begin{equation*}
      \hat{\bfalpha}-\bfalpha_0 \to 0,
      \quad \hat{\bfbeta}-\bfbeta_0 \to 0,
      \quad \sup_{t\in[0,\tau]}|e^{-\hat{\Lambda}(t)}-e^{-\Lambda_0(t)}| \to 0
      \quad a.s.. 
    \end{equation*}
\end{theorem}

The proof follows the general framework in \cite{Murphy94}.
The estimator for the baseline hazard satisfies the equation
\begin{equation}\label{def:Lam hat}
\hat{\Lambda}(t)=\int_0^t \left\{\sum_{i=1}^n W_i^{\hat{\bfth}}(u)e^{\hat{\bfbeta}^\top\vecZ_i}\right\}^{-1}d \bar{N}(u),
\end{equation}
where
\begin{equation}\label{def:W}
W_i^{\bfth}(t)=\big\{\delta^1_i+\delta^c_i \phi_i^{\bfth}(X_i)\big\} I\{t \le X_i\}
- \phi_i^{\bfth}(Q_i) I\{t \le Q_i\},
\end{equation}
and $\phi_i^{\bfth}(\cdot) $ is given in \eqref{def:phi}.
A bridge between $\hat{\Lambda}$ and $\Lambda_0$ is constructed as
\begin{equation}\label{def:Lam bar}
\bar{\Lambda}(t)= \int_0^t \left\{\sum_{i=1}^n \phi^{\bfth_0}_i(u) Y_i(u)  e^{\bfbeta_0^\top\vecZ_i}\right\}^{-1}d \bar{N}(u).
\end{equation}

The details of the proof deserve some extra comments here, as it achieves the a.s. convergence with
    a baseline hazard unbounded in its support using a few innovative steps.
    First, we apply Helly's selection theorem to the C\`{a}dl\`{a}g function sequence $e^{-\hat{\Lambda}}$.
    Then, the upper bound for $\hat{\Lambda}$ in any interval $[0,\tau^*]\subset(0,\tau)$ is
    established via the lower bound for $n^{-1}\sum_{i=1}^nW_i^{\hat{\bfth}}(u)e^{\hat{\bfbeta}^\top\vecZ_i}$.
    We manage to show that
        the ratio $\gamma(t)=d\hat{\Lambda}(t)/d\bar{\Lambda}(t)$ is bounded between zero and infinity
        for  all $t \in (0,\tau)$
        despite the indefinite quotient at $0$ and $\tau$.
    Finally, we conclude the proof by showing that $\gamma(t)=1$ using an identifiability argument.

For the purposes of the asymptotic normality below,
we have a similar result: 
\begin{customthm}{\ref{thm:consistency}'}\label{thm:consistency'}
Under Assumptions \ref{assumpt:parameter}, \ref{assumpt:cov bound}, \ref{assumpt:baseline'} and \ref{assumpt:ltrc},
    the NPMLE estimator for
    $L^I$ defined later in \eqref{Lik:obsIC}, $\hat{\bfth}=(\hat{\bfalpha},\hat{\bfbeta},\hat{\Lambda}(\cdot))$,  is consistent.
    That is
    \begin{equation}\label{def:theta norm}
      \hat{\bfalpha}-\bfalpha_0 \to 0,
      \quad \hat{\bfbeta}-\bfbeta_0 \to 0,
      \quad \sup_{t\in[0,\tau']}|\hat{\Lambda}(t)-\Lambda_0(t)| \to 0
      \quad a.s.. 
    \end{equation}
\end{customthm}

\subsection{Asymptotic Normality of NPMLE}

The divergence of the cumulative baseline hazard $\Lambda_0$ at $\tau$ eventually becomes an obstacle in the study of weak convergence. It is involved in all the second order terms including both the parametric parts and the nonparametric part. Existing techniques, mostly relying on a finite upper bound of $\Lambda_0$, cannot deal with it.
To proceed with the theoretical endeavor, we avoid the divergent tail by slightly modifying the likelihood.
That is, we make an interval censoring window $(\tau',\tau)$  
        close to the end of study,   so that the failure indicator $A$ is always observed for those at-risk at time $\tau'$,
     but their failure times are unknown if $A=1$.
     We note that this is for technical reason only, so that the baseline cumulative hazard is always bounded
        at the observed failure times as $n\to\infty$. In practical applications this modification of the likelihood is unnecessary since the observed SAB events are recorded in dates, so that there is always at least one day gap between when a (possibly censored) SAB event can happen and when a woman is considered cured.
        
    Let $\delta^\tau=\etaA \cdot I(X > \tau')$ be the interval-censoring indicator in  $(\tau',\tau)$.
    Notice that $S(t)-S(\tau)=S(t)$ for any $t<\tau$. We have the resulting interval-censored data likelihood that is modified from \eqref{Lik:obs}:
    \begin{equation}\label{Lik:obsIC}
      L^{I}(\bfth)=\prod_{i=1}^n\frac{\big\{p\lambda_i(X_i)e^{\bfbeta^\top\vecZ_i}S_i(X_i)\big\}^{\delta^1_i}
        (1-p_i)^{\delta^0_i}\big\{1-p_i+p_iS_i(X_i)\big\}^{\delta^c_i}\{p_iS_i(\tau')\}^{\delta^\tau_i} }{1-p_i+p_iS_i(Q_i)}.
    \end{equation}
The corresponding  log-likelihood $l^I_n=\log L^I$ is
    \begin{align} \label{loglik:IC}
    \notag l^I_n=\sum_{i=1}^n &
\int_0^{\tau'} \log\Big(\phi_i^{\bfth}(u)e^{\bfbeta^\top\vecZ_i}\triangle \Lambda(u)\Big)dN_i(u)
-\int_0^{\tau'} Y_i(u)\phi_i^{\bfth}(u)e^{\bfbeta^\top\vecZ_i}d\Lambda(u) \\
&+\{N_i(\tau)-N_i(\tau')\}\log\phi_i^{\bfth}(\tau')+Y_i(\tau) \log(1-\phi_i^{\bfth}(\tau')).
    \end{align}
The proof then follows the framework in \cite{Murphy95}
    to verify the conditions of Theorem 3.3.1 from \cite{vandervaart1996}. 
We shall describe the functional space in which weak convergence is established. 
Let $H_\infty$ be the space containing elements in the form of $\bfh = (\bfa,\bfb,\hbaseline)$, where
the vectors $\bfa$ and $\bfb$ are of the same dimensions as $\bfalpha$ and $\bfbeta$, respectively, and the
function $\hbaseline(\cdot)$ is defined on $[0,\tau']$ with $\hbaseline(0)=0$ and is of bounded variation, i.e. the total variation of $\hbaseline$ over $[0,\tau']$,
$$
V_0^{\tau'} \hbaseline =  \sup_{\stackrel{0=u_0<\dots<u_s=\tau'}{s = 1, 2, \dots}} 
\sum_{j=1}^s |\hbaseline(u_j) - \hbaseline(u_{j-1})|
$$
is finite. 
Define a norm $\| \cdot \|_H$ on $H_\infty$:
\begin{equation*}
\|(\bfa,\bfb,\hbaseline)\|_H=\|\bfa\|_1 + \|\bfb\|_1 +  V_0^{\tau'}|\hbaseline|,
\end{equation*}
and spaces indexed by a positive real number $p$
\begin{equation*}
    H_p=\left\{\bfh  : \|\bfh\|_H<p\right\}. 
\end{equation*}
For each $p$, define $l^\infty(H_p)$ as the functional space of all uniformly bounded linear map $H_p\mapsto \R$, i.e.
$$
\forall \Psi \in l^\infty(H_p), \; \sup_{\bfh \in H_p} |\Psi(\bfh)| < \infty. 
$$
The parameter $\bfth=(\bfalpha,\bfbeta,\Lambda)$ as a function in $l^\infty(H_p)$ is defined as
\begin{equation*}
\bfth(\bfh)=\bfa^\top\bfalpha+\bfb^\top\bfbeta+\int_0^{\tau'} \hbaseline(u)d\Lambda(u).
\end{equation*}
The induced functional norm 
is 
equivalent to the norm in \eqref{def:theta norm} where consistency (Theorem \ref{thm:consistency'}) is established;  
 we denote $\|\bfth\|$. 

\begin{theorem}\label{thm:normal}
Let $\hat{\bfth}=(\hat{\bfalpha},\hat{\bfbeta},\hat{\Lambda}_0(\cdot))$ be the NPMLE
 for the log-likelihood $l^I_n$ in \eqref{loglik:IC}.
  Under Assumptions \ref{assumpt:parameter}, \ref{assumpt:cov bound},
    \ref{assumpt:baseline'} and \ref{assumpt:ltrc},
  \begin{equation*}
    \sqrt{n}(\hat{\bfth}-\bfth_0)
    \longrightarrow \mathcal{G}, \text{ in } l^\infty(H_p)
  \end{equation*}
  weakly for a tight Gaussian process $\mathcal{G}$ on $l^\infty(H_p)$ with covariance process
  \begin{equation*}
    \Cov (\mathcal{G}(\bfh),\mathcal{G}(\bfh^*))=\bfa^\top\bfsig^{-1}_a(\bfh^*)
        +\bfb^\top\bfsig^{-1}_b(\bfh^*)+\int_0^{\tau'}\hbaseline(u)\sigma^{-1}_\hbaseline(\bfh^*) (u)d\Lambda_0(u),
  \end{equation*}
  where $\bfh = (\bfa,\bfb,\hbaseline)$, and $\sigma(\bfh)=\Big(\bfsig_a(\bfh),\bfsig_b(\bfh),\sigma_\hbaseline(\bfh)\Big)$ is given in the Appendix \eqref{def:sig}.
\end{theorem}

Let $\hat{\sigma}$ be a nature estimator for the operator $\sigma$
    by substituting the true parameter $\bfth_0$ and expectation 
    with the estimator $\hat{\bfth}$ and the sample average. 
\begin{theorem}\label{thm:var est}
  Under Assumptions \ref{assumpt:parameter}, \ref{assumpt:cov bound},
    \ref{assumpt:baseline'} and \ref{assumpt:ltrc}, $\hat{\sigma}$ is asymptotically equivalent to the information matrix in \eqref{eq:louis}.
    The solution to $\mathbf{g}=\hat{\sigma}^{-1}(\bfh)$ exists with probability going to $1$
        as $n$ increases and
    \begin{equation}\label{def:natural var}
   \notag    \bfa^\top\hat{\bfsig}^{-1}_a(\bfh^*)
        +\bfb^\top\hat{\bfsig}^{-1}_b(\bfh^*)
        +\int_0^{\tau'}\hbaseline(u)\hat{\sigma}^{-1}_\hbaseline(\bfh^*)(u)d\hat{\Lambda}(u) \\
    \stackrel{\P}{\longrightarrow}
    \Cov (\mathcal{G}(\bfh),\mathcal{G}(\bfh^*)).
    \end{equation}

\end{theorem}

\section{Simulation study} \label{section:simulation}

\subsection{Simulation setup}

Here we detail our data simulation procedure for all of the simulation studies.
Simulating
cure-rate model data presents its own challenges. To be comparable with the spontaneous
abortion data which we examine in the next section, we consider finite time $\tau$, which is set to be 20
(weeks). The covariates are the same for the logistic and the Cox part of the regression models
and, unless otherwise specified, consisting of $Z_1 \sim N(4, 1)$, with corresponding parameters
$(\alpha_1,\beta_1)$,
and $Z_2 \sim$ Bernouilli $(p = 0.3)$, with corresponding parameters $(\alpha_2,\beta_2)$.
The logistic
regression part also includes an intercept $\alpha_0$.

We begin by generating a larger sample than we desire to account for those who will be
left out due to truncation. Values for $\bfalpha$ are chosen to procure the desired percentage of cured
individuals on average in the population, and we refer to this as the \% of cured individuals
in a simulation study. An individual is designated as either cured or not with the
probability determined from the logistic model.

The baseline survival function for the Cox model is set as $S_0(t)=20-t$,
the survival function of a Uniform $(0,20)$ random variable.
The baseline cumulative hazard is thus $\Lambda_0(t)=20\{1-e^{-t}\}$.
For those not cured individuals we  generate an event time 
$T=20\{1-U^{\exp(-\beta_1Z_1-\beta_2Z_2)}\}$, where $U\sim$ Uniform $(0,1)$.

Truncation times are generated from Uniform $(0, a)$ for some $a < 15$ chosen so that on
average the desired percentage of uncured individuals are truncated out. We refer to this
percentage as the \% of truncation. Once the truncation times are generated, all individuals
with event times less than their truncation times are removed, and we reduce the data set
to the desired sample size by taking the first $n$ individuals from those who remain.
Finally, when there is censoring the censoring times are generated from Uniform $(15,b)$
for some $b > 20$ so that on average the desired percentage of the $n$ individuals
(including those who are cured) will have a censoring time less than $\min(T_i, 20)$. We refer to
this percentage as the \% of censoring.
We ran all simulations with 500 trials below.

\subsection{Simulation results}

In Tables \ref{table:simulation} we examine the performance of the NPMLE.
We consider a smaller sample
size n = 200 and a larger sample size n = 1000, and like in the pregnancy studies for SAB
we assume that a majority 75\% of the subjects are cured. We ran simulations over the
combination of two truncation scenarios (10\%, 20\%) and two censoring scenarios (0\%, 20\%).
In the tables we provide the average parameter estimates (``Estimate"), the sample standard
deviation of these estimates over the 500 simulation trials (``Sample SD"), the mean over
the 500 trials of the standard errors based on our variance estimation (``SE"), and the
empirical coverage probabilities (``Coverage") of the nominal 95\% confidence intervals using
the SE's.

According to the table, the performance of NPMLE is quite good. The average estimates
of the parameters are generally close to their true values in all
scenarios. This includes for the Cox part of the model under the smaller sample size n = 200, where  only
 about 25\% of the sample have events when there is no censoring, and even few in the presence of censoring.
The variance estimator generally improves with larger sample size, especially for the Cox part of the model and with 20\% censoring, which also reflects in the coverage probabilities of the nominal 95\% confidence intervals. Note that with 500 simulation trials these empirical coverage probabilities have about $\pm$2\% margin of error.

\section{Analysis of spontaneous abortion data}\label{section:SAB}

The data we investigate come from the OTIS autoimmune disease in pregnancy database as mentioned
earlier. Our sample includes pregnant women who entered a research study between 2005
and 2012. It consists of n = 929 women who entered the study before week 20 of their
gestation,  with complete covariate information.
Among them 482 (52\%) were pregnant women with certain autoimmune diseases
who were treated with medications under investigation, 265 (28\%) were women with the same specific
autoimmune diseases who were not treated with the  medications under investigation, and the rest
182 (20\%) were healthy pregnant women without autoimmune diseases who were not treated
with the  medications. \cite{Chambers01}  discussed the importance of
having a diseased control group, since some of the adverse outcomes in pregnancy may be
due to the diseases instead of the medications. There were a total of 66 SAB events, and 2
women were lost to follow up before 20 weeks of gestation.

There are a number of risk factors for spontaneous abortion that have been identified in
the literature \cite[for example]{Chambers13}.
Alternatively, we can use a data driven selection method for risk factors in our cure rate model.
For each baseline covariate, we use the Wald test with two degrees of freedom for both coefficients in the logistic and the Cox part of the model.
We first screen the covariates with a univariate cure rate model, with a $p$-value cutoff of 0.2 for the Wald test.
We then run a backward selection, with a $p$-value cutoff of 0.1 for the Wald test.
The selected variables are body mass index (BMI) group (0: {BMI} $<$ 18.5, 1: {BMI} $\in$ [18.5,24.9), 2: {BMI} $\in$ [25,29.9]), 3: {BMI} $>$ 30) treated as numerical due to small number of total SAB events,
 gravidity 
 $>$ 1 or not, i.e.~whether a woman had been previously pregnant, whether there was smoking (Y/N) or alchohol (Y/N) intake during early pregnancy.
 We fit our final cure  model to the data with these covariates and exposure status, and the results using
the NPMLE are given in Table \ref{table:SAB} left columns.

From Table \ref{table:SAB}, we see that larger body mass index significantly decreases the probability
of SAB in the logistic part of the model. The probability of
SAB of either healthy control group or disease control group is not significantly different from the medication exposed women.
The Cox regression part of the model identified all four covariates as significant factors for the hazard
of SAB. In the cure model context since the Cox model is only used for those who eventually have events (observed or censored), this part of the model  should be understood as impact of the covariates on the timing of SAB; that is, significantly later timing of
SAB for those who had larger body mass index, gravidity $>$ 1 or smoking,
and significantly earlier timing for those who had alcohol. Figure \ref{fig:KMbycov} illustrates the significance (or not) of BMI and alcohol in association with the overall risk of SAB by 20 weeks of gestation as well as with the timing of SAB among those women not observed to be cured.

Acounting for the left truncation, classical survival analysis methods including the Cox proportional hazards regression model have been advocated in the
literature  \cite[]{Meister08, Xu11}. As a comparison,  Table \ref{table:SAB} right columns (lower half) show the results of the the classic Cox regression model fitted to the data by
treating all the cured individuals as right-censored at 20 weeks of gestation, as is currently done in the practical  analysis of SAB data \cite[]{Chambers13}.
Gravidity and smoking  are no longer significant predictors of SAB.
Note that under the proportional hazards assumption, nonsignificant effects of gravidity or smoking translates to no significant differences in the cumulative
risks of SAB; that is, the impact on the timing of SAB is no longer distinguished from the impact on the overall cumulative risk of SAB (Y/N) by 20 weeks of gestation.
In addition, as mentioned before, 
 treating the majority of the women (who did not have SAB) as right-censored can  lead to
substantial loss of information.

Finally we also fit the `naive' logistic regression model alone to the data, using whether a woman has SAB (Y/N) as the outcome. 
The results are also given in the right columns (upper half) of Table \ref{table:SAB}. We note that this model does not properly handle left truncation, and results are wildly different from the other model fits and should be not trusted.

\section{Discussion and Conclusion}

In this paper we have developed an NPMLE approach to fit the  mixture type cure rate
models to data with left truncation in addition to right-censoring.
As illustrated in the data analysis, the cure rate model methodology
 developed here is able to make use of the information from both the women who had
SAB and those who were observed not to have SAB, as well as to separate the differential
regression effects of the covariates on both the cumulative risk of SAB as well as the timing
of it among those who experience SAB. We anticipate this methodology to impact the
practical analysis of pregnancy and other similar types of data.
An `alpha' version of a corresponding R package is currently being tested internally.

Different from the usual cure rate data where the long-term survivors are always right-censored,
in our pregnancy studies we observe the majority of the `cured' women. This greatly
improves the practical identifiability of the cured portion (Sy and Taylor, 2000; Lu and Ying,
2004), as well as substantially increases the amount of information available for estimating the
model parameters. Our inference procedures utilize the NPMLE,
    together with the ``ghost copy'' EM algorithm to produce estimators for the model parameters.
The variances of the estimators can be obtained in
closed form using the \cite{Louis82} formula.
In our simulations, the variance estimator leads to relatively accurate coverage of the 95\% confidence intervals.

In our proof for consistency, we have worked through an unbounded cumulative baseline hazard,
which has rarely been discussed in existing literatures.
Ideally, we would like to show asymptotic normality without assuming the interval-censoring tail window.
However, the weak convergence of nonparametric estimators often requires a stronger set of assumptions.
As a result, the unbounded $\Lambda_0$ in the log-likelihood causes trouble in the Fr\'{e}chet differentiability
and continuously invertibility steps.
The ``chop-off" argument applied in consistency does not work here as $\Lambda_0$ 
appear    in both the parametric part and the nonparametric part of the directional score.

Finally for left truncated data much work has been done recently under the length-biased
assumption \cite[among others]{Asgharian06,Ning10,Qin11}. For
enrollment into observational pregnancy studies like ours, we do not think that the uniform
distributional assumption necessarily holds, as is evident in Figure \ref{fig:Qhist}. However, it
would be of interest to compare the efficiency (as well as bias) of the different approaches,
and to develop methods under other parametric assumptions that are more suitable for the
entry times to pregnancy studies.


\bibliographystyle{natbib}
\setlength{\bibhang}{0pt}
\bibliography{Hou_reference}

\newpage
\begin{appendix}

\setlist[enumerate]{leftmargin=0ex,
    itemindent=4ex, label = \textbf{\alph*)}}

\renewcommand\theequation{\thesection\arabic{equation}}
\renewcommand\thelemma{\thesection\arabic{lemma}}
\renewcommand\thetheorem{\thesection\arabic{theorem}}

\setcounter{equation}{0}
\setcounter{lemma}{0}
\setcounter{theorem}{0}
\section{Appendix A: Proofs}\label{appendixA}

\subsection{The Existence of NPMLE }


 \begin{proof}[\underline{Proof of Theorem \ref{thm:exist}}]
 Let $\theta_B$ be the maximizer on the compliment of compact set
    $\{\|\bfalpha\| \vee \|\bfbeta\|  \vee \|\blam\|  \le B\}$. 
We show that $l(\theta_B) \to -\infty$ when $B \to \infty$.

By Assumptions \ref{assumpt:parameter} and \ref{assumpt:cov bound},
we have the bound \eqref{def:Mm}.

All terms in the log-likelihood are bounded except for
$$\sum_{i=1}^{n}\Big\{\delta^1_i\log \lambda(X_i)-\delta^1_i e^{\boldsymbol{\beta}^\top\vecZ_i} \Lambda(X_i)\Big\}. $$
Let $\lambda_{\max}$ be the largest element in $\blam$. The expression above has the upper bound
$$\log( \lambda_{\max}/\upM)- \lambda_{\max}/\upM -K\log \upM,  $$ which
diverges to $-\infty$ when we set $B \to \infty$.

Then, the global maximizer must be in one of the compact set $\{\|\bfalpha\| \vee \|\bfbeta\|  \vee \|\blam\|  \le B^*\}$ for some $B^*>0$. 
 \end{proof}

Let $W_i^{\bfth}(t)$ be defined as in \eqref{def:W}. 
We define a generic inequality to be referenced  later, for any $\bfth = (\bfalpha,\bfbeta, \Lambda)$ in the parameter space whose baseline cumulative hazard $\Lambda$ is a step function jumping only at the observed event times, $t_1, \dots, t_K$:
\begin{equation}\label{eqn:EMjump}
   	0 < d\Lambda(t_k) \le \left(\sum_{j=1}^nW_j^{\bfth}	(t_k)e^{\bfbeta^\top\vecZ_j}\right)^{-1}d\bar{N}(t_k), \quad k = 1, \dots, K.
\end{equation}
The conclusion of the following Lemma is used in the proofs of both Lemma \ref{lemma:EM} and Theorem \ref{thm:consistency}. 

\begin{lemma}\label{lem:bound baseline} 
    Let $\bfth_{(n)} = \left(\bfalpha_{(n)}, \bfbeta_{(n)}, \Lambda_{(n)}\right)$ be a sequence in the  parameter space where $\Lambda_{(n)}$ is a non-decreasing step function with jumps only at the observed event times. Suppose that $\bfth_{(n)}$ satisfies \eqref{eqn:EMjump} and has a subsequence
$\bfth_{(n_k)}$ converging to a limiting point ${\bfth}^* = (\bfalpha^*, \bfbeta^*, \Lambda^*)$ a.s.:
   \begin{equation}\label{eqn:Helly}
  \bfalpha_{(n_k)}-\bfalpha^* \to 0, \quad 
   \bfbeta_{(n_k)}-\bfbeta^* \to 0, \quad 
 \sup_{t\in [0,\tau]}|e^{-\Lambda_{(n_k)}(t)}-e^{-\Lambda^*(t)}| \to 0, \quad a.s..
\end{equation}
    Under Assumptions \ref{assumpt:parameter} - \ref{assumpt:ltrc}, 
\begin{enumerate}
\item \label{lempart:finiteLam}
$\Lambda^*(t) < \infty \text{ for all } t<\tau$;
\item \label{lempart:Whead}
$\inf_{t\in [0,\zeta]}\E [W^{\bfth^*}(t)e^{\bfbeta^{*\top}\vecZ}]>C_w, \text{ for some } C_w>0$.
\end{enumerate}
\end{lemma}

\begin{proof}[\underline{Proof of Lemma \ref{lem:bound baseline}}] 
By checking the uniform continuity of $W_i^{\bfth}(t)e^{\bfbeta^\top\vecZ_i}$ in $(\bfalpha,\bfbeta,e^{-\Lambda(t)})$,
we may establish
\begin{equation*}
  \sup_{t \in [0,\tau]}
\left|W_i^{\bfth^*}(t)e^{\bfbeta^{*\top}\vecZ_i}-
W_i^{\bfth_{(n_k)}}(t)e^{\bfbeta_{(n_k)}^\top\vecZ_i}\right|
\to 0, \quad a.s..
\end{equation*}
$W_i^{\bfth}(t)$ as a function of observed random variables belongs to
a Glivenko-Cantelli class of uniformly bounded functions with uniformly bounded variation.
Thus, the pointwise convergence can be strengthen to be uniform convergence,
\begin{equation*}
  \sup_{t \in [0,\tau]}
\left|\frac{1}{n}\sum_{i=1}^{n_k} W_i^{\bfth_{(n_k)}}(t)e^{\bfbeta_{(n_k)}^\top\vecZ_i}
-\E \left[ W^{\bfth^*}(t)e^{\bfbeta^{*\top}\vecZ}\right]\right|
\asto 0.
\end{equation*}
Note that $n^{-1}\sum_{i=1}^{n_k}W_i^{\bfth_{(n_k)}}(t)e^{\bfbeta_{(n_k)}^\top\vecZ_i}$ is c\`{a}gl\`{a}d,
so its limit $\E  [W^{\bfth^*}(t)e^{\bfbeta^{*\top}\vecZ} ]$ must also be c\`{a}gl\`{a}d.

\begin{enumerate}
   \item
Let $\tau^*=\inf\{t\in [0,\zeta]: e^{-\Lambda^*(t)}=0\}$. 
We shall prove that $\tau^*=\tau$.

Suppose that $\tau^*$ is an interior point of $[0,\tau]$.
From Assumption \ref{assumpt:ltrc}, $d\Lambda_0([s,t]) = \Lambda_0(t)  -\Lambda_0(s) >0$ for any $s<t$ in $[0,\tau]$. 
By the definition of $\tau^*$, $\Lambda^*(t)=\infty$ and $\phi^{\bfth^*}(t) = 0$ for $t \in [\tau^*,\tau]$, so we have
\begin{equation*}
\E \left[W^{\bfth^*}(\tau^*)e^{\bfbeta^{*\top}\vecZ}\right]=\E \left[ \int_{\tau^*_-}^\tau e^{\bfbeta^{*\top}\vecZ} dN(u) \right]>0.
\end{equation*}
By the left continuity of $W_i^{\bfth}(t)$, $\exists \ s < \tau^*$, s.t.
\begin{equation*}
\inf_{t\in [s,\tau^*]}\E \left[ W^{\bfth^*}(t)  e^{\bfbeta^{*\top}\vecZ} \right] \ge \frac{1}{2}\E \left[ \int_{\tau^*_-}^\tau e^{\bfbeta^{*\top}\vecZ} dN(u) \right].
\end{equation*}
The total increment of $\Lambda_{(n_k)}$ in $[s,\tau^*]$ must be bounded almost surely according to \eqref{eqn:EMjump}. By the definition of $\tau^*$, $\Lambda^*(s)<\infty$. Putting these together, we reach the contradiction,
\begin{align*}
\Lambda^*(\tau^*) \le  \varlimsup_{k \to\infty}\Lambda_{(n_k)}(\tau^*)
\le & \varlimsup_{k \to\infty}\Lambda_{(n_k)}(s)+
	\int_{s_+}^{\tau^*} \frac{d \bar{N}(u)}
	{\sum_{i=1}^{n_k}W^{\bfth_{(n_k)}}_i(u)e^{\bfbeta_{(n_k)}^\top\vecZ_i}} \\
	\le & \Lambda^*(s)+ \frac{\tau^*-s}{\inf_{t\in [s,\tau^*]}\E [W^{\bfth^*}(t)  e^{\bfbeta^{*\top}\vecZ}]}<\infty .
\end{align*}

The other case is $\tau^* = 0$. Then, $\Lambda^*(t)=\infty$ and $\phi^{\bfth^*}(t) = 0$ for $t \in [0,\tau]$. The contradiction is easily established as
$$
\E \left[W^{\bfth^*}(0)e^{\bfbeta^{*\top}\vecZ}\right]=\E \left[ \int_0^\tau e^{\bfbeta^{*\top}\vecZ} dN(u) \right]>0.
$$

\item Since $\E [W^{\bfth^*}(t)  e^{\bfbeta^{*\top}\vecZ}]$ is c\`{a}gl\`{a}d, $\bfth_{(n_k)}$ satisfies \eqref{eqn:EMjump} and converges uniformly to $\bfth^*$, it can be seen that $\E [W^{\bfth^*}(t)  e^{\bfbeta^{*\top}\vecZ}] \geq 0$
	 over the interior of $[0,\zeta]$.

Write $n_k^{-1}\sum_{i=1}^{n_k}W_i^{\bfth}(t)e^{\bfbeta^\top\vecZ_i}$ as 
\begin{align}
\notag &n_k^{-1}\sum_{i=1}^{n_k} \int_{t-}^\tau \big\{1-\phi_i^{\bfth}(u)\big\}e^{\bfbeta^\top\vecZ_i}dN_i(u)
        +\int_t^\tau Y_i(u)e^{\bfbeta^\top\vecZ_i} d \phi_i^{\bfth}(u)+Y_i(t)\phi_i^{\bfth}(t)e^{\bfbeta^\top\vecZ_i} \\
 \notag =&n_k^{-1}\sum_{i=1}^{n_k}  \int_{t+}^\tau \left[1-\phi_i^{\bfth}(u)
    -\frac{\sum_{j=1}^{n_k}Y_j(u)\phi_j^{\bfth}(u)\big\{1-\phi_j^{\bfth}(u)\big\}e^{\bfbeta^\top\vecZ_j}}
        {\sum_{j=1}^{n_k}W_j^{\bfth}(u)e^{\bfbeta^\top\vecZ_j}}\right]e^{\bfbeta^\top\vecZ_i}dN_i(u)\\
 & \quad + \big\{1-\phi_i^{\bfth}(t)\big\}e^{\bfbeta^\top\vecZ_i}dN_i(t)+ Y_i(t)\phi_i^{\bfth}(t)e^{\bfbeta^\top\vecZ_i} \label{eqn:WasdN} .
\end{align}	
By Assumption \ref{assumpt:ltrc}, all $Q_i < \zeta$ a.s..
Thus,
\begin{align*}
\E \left[W^{\bfth^*}(\zeta)  e^{\bfbeta^{*\top}\vecZ}\right]=&\E \left[ \big\{\delta^1+\delta^c \phi^{\bfth^*}(X)\big\} I\{\zeta \le X\}e^{\bfbeta^{*\top}\vecZ} \right] \notag \\
\ge & \E \left[ \int_\zeta^\tau e^{\bfbeta^{*\top}\vecZ}dN(u) \right] >0.
\end{align*}
For $t<\zeta$, the difference $\E [W^{\bfth^*}(t)e^{\bfbeta^{*\top}\vecZ}]-\E[ W^{\bfth^*}(\zeta)e^{\bfbeta^{*\top}\vecZ}]$ is the limit of an integral  like that in \eqref{eqn:WasdN}, where the integrand has $\sum_{j=1}^{n_k}W_j^{\bfth}(u)e^{\bfbeta^\top\vecZ_j}$ in the denominator. So it has potential singularities at the zeros of
$\E[ W^{\bfth^*}(u)e^{\bfbeta^{*\top}\vecZ}]$ for $ u \in [t,\zeta]$.
We shall show that $\E [W^{\bfth^*}(u)e^{\bfbeta^{*\top}\vecZ}]$ is differentiable with respect
to $d\Lambda_0(u)$ in $[0,\zeta]$, so that its zero $u_0$ leads to the divergent form
$
- \int_t ^\zeta |u-u_0|^{-1} du. 
$
We will then reach the contradiction that $\E [W^{\bfth^*}(t)e^{\bfbeta^{*\top}\vecZ}]=-\infty$, as seen below.

Denote $R_0$  the set of zeros and limiting zeros from right for
$\E[ W^{\bfth^*}(u)e^{\bfbeta^{*\top}\vecZ}]$.
Let set $R_{\triangle u}$ be the $\triangle u$ neighborhood of $R_0$ and
$\Omega^t_{\triangle u}=[t,\zeta]  \setminus R_{\triangle u}$.
$\E [W^{\bfth^*}(u)e^{\bfbeta^{*\top}\vecZ}]$ is bounded away from zero on
$\Omega^t_{\triangle u}$. Through \eqref{eqn:WasdN},
\begin{align}
\notag \E &\left[W^{\bfth^*}(t)e^{\bfbeta^{*\top}\vecZ}\right]-\E \left[W^{\bfth^*}(\zeta)e^{\bfbeta^{*\top}\vecZ}\right]\\
\notag&\le
     -\int_{\Omega^t_{\triangle u}}\frac{\E [Y(u)\phi^{\bfth^*}(u)\big\{1-\phi^{\bfth^*}(u)\big\}e^{\bfbeta^{*\top}\vecZ}]}
        {\E[ W^{\bfth^*}(u)e^{\bfbeta^{*\top}\vecZ}]}\E \left[e^{\bfbeta^{*\top}\vecZ}dN(u)\right]\\
 & +  \E \left[ \int_{t+}^\zeta \{1-\phi^{\bfth^*}(u)\}dN(u) + \big\{1-\phi^{\bfth^*}(t)\big\}e^{\bfbeta^{*\top}\vecZ}dN(t)+ Y(t)\phi^{\bfth^*}(t)e^{\bfbeta^{*\top}\vecZ}\right].
 \label{eqn:Wdiff}
\end{align}

From part \ref{lempart:finiteLam}, $e^{-\Lambda^*(\zeta)}>0$. For any $ u<\zeta$,
\begin{equation*}
\phi^{\bfth^*}_i(u) \ge \phi^{\bfth^*}_i(\zeta) \ge \frac{\lowm e^{-\upM\Lambda^*(\zeta)}}{1+\lowm e^{-\upM\Lambda^*(\zeta)}}>0.
\end{equation*}
So the limit of numerator term $\E [Y(u)\phi^{\bfth^*}(u)\{1-\phi^{\bfth^*}(u)\}e^{\bfbeta^{*\top}\vecZ}]$
is bounded away from zero. And  $\forall u \in [0,\zeta]$,
\begin{align*}
\left|\frac{d\E W^{\bfth^*}(u)}{d\Lambda_0(u)}\right|
=&\left| \E \left[ \big\{1-\phi^{\bfth^*}(u)\big\}Y(u)\phi^{\bfth_0}(u)
	e^{\bfbeta_0^\top\vecZ}
	-\phi^{\bfth^*}(u)\frac{d\E[Y(u)|\vecA,\vecZ]}{d\Lambda_0(u)}\right]\right|\\
	\le & \upM+\mathcal{L} <\infty.
\end{align*}
The first term in \eqref{eqn:Wdiff} diverges to $-\infty$ when $\triangle u \to 0$.
The other terms are bounded, so this is the desired contradiction.
\end{enumerate}
\end{proof}

\begin{proof}[\underline{Proof of Lemma \ref{lemma:EM}}]  
For any $\bfth$ denote $\lambda_{\max,\zeta}=\max\{\lambda_k: t_k \le \zeta \}$, where $\zeta$ is the upper bound of truncation time defined in Assumption \ref{assumpt:ltrc}. 
 Define a set in the parameter space:
\begin{equation}
{\Theta} = \left\{\bfth=(\bfalpha,\bfbeta,\Lambda) | 
\lambda_{\max,\zeta} \le  n^{-1}2/C_w\right\} ,
\end{equation}
with $C_w$ defined in Lemma \ref{lem:bound baseline}. 
We would like to show that 
\begin{equation}\label{eqn:Omega0_to_1}
\lim_{n\to \infty}\P(\hat{\bfth}, \tilde{\bfth}\in \widehat{\Theta}) = 1.
\end{equation}
This is done through applying Lemma \ref{lem:bound baseline}, so we will need to verify condition \eqref{eqn:EMjump}  for  $\tilde{\bfth}$ and $\hat{\bfth}$. 

First, define the marginal of the complete data likelihood
\begin{align*}
     \tilde{L}(\bfth)=&\sum_{A_i=0,1}\sum_{M_i=0}^\infty \sum_{\Tghost_{i1}=t_k: t_k\le Q_i}\dots\sum_{\Tghost_{iM_i}=t_k: t_k\le Q_i} L^c_i(\bfth) \notag \\
= &\prod_{i=1}^n\frac{\big\{p_i\lambda_i(X_i)S_i(X_i)\big\}^{\delta^1_i}
        (1-p_i)^{\delta^0_i}\big\{p_iS_i(X_i)+1-p_i\big\}^{\delta^c_i}}
        {1-p_i\sum_{k: t_k\le Q_i}\lambda_k \phazi S_i(t_k)}.
   \end{align*}
From \eqref{Lik:complete} it can be seen that the complete data likelihood $L^c(\bfth)$ can be decomposed into the product
of one logistic part concave in $\bfalpha$ with one Cox part concave in $(\bfbeta,\blam)$.
Thus, it is concave in $\bfth$.
$\tilde{L}(\bfth)$ is also concave as the sum over concave functions. 

Next we show that the EM finds the unique stationary point of $\tilde{L}(\bfth)$, which then must  be the global maximizer since it is concave.  
Consider the conditional expectation given the observed data as in \eqref{EM:condE1} - \eqref{EM:condE3}.
It can be verified directly (we skip the algebraic details here) that:
\begin{equation*}
\nabla \log \tilde{L}(\bfth)=\E_{\bfth}[ \nabla \log L^c(\bfth) |\Ocal ].
\end{equation*}
The estimator $\tilde{\bfth}$ is by definition the solution to the left-hand side of the above being zero, hence also the stationary point of $\tilde{L}(\bfth)$. 

We write down the stationary equation
$\bfth^{(l)}=\bfth^{(l+1)} = \tilde{\bfth}$ for $\tilde\lambda_k$'s at convergence,
\begin{equation*}
\tilde{\lambda}_k=\frac{1+\tilde{\lambda}_k\sum_{i=1}^n\frac{\tilde{p}_ie^{\tilde{\bfbeta}\top\vecZ_i}
\tilde{S}_i(t_k)I(Q_i \ge t_k)}
{1-\tilde{p}_i\sum_{h:h<Q_i}\tilde{f}_i(t_h)}}
	{\sum_{i=1}^n\left\{\delta^1_i I(X_i \ge t_k)+\delta^c_i\phi^{\tilde{\bfth}}_i(X_i)I(X_i \ge t_k)
	+\sum_{j\ge k}\frac{\tilde{p}_i \tilde{f}_i(t_j)I(Q_i \ge t_j)}
{1-\tilde{p}_i\sum_{h:h<Q_i}\tilde{f}_i(t_h)}
	\right\}
	e^{\tilde{\bfbeta}\top\vecZ_i}},
\end{equation*}
where $f_i$ was previously defined just above \eqref{dist:k}.
Combining $\tilde{\lambda}_k$ terms  leads to
\begin{align}
 \notag
\tilde{\lambda}_k^{-1}=\sum_{i=1}^n&\bigg\{\delta^1_i I(X_i \ge t_k)+\delta^c_i\phi^{\tilde{\bfth}}_i(X_i)I(X_i \ge t_k)\\
	&-\tilde{p}_i \frac{\tilde{S}_i(t_k)I(Q_i \ge t_k)-\sum_{j\ge k}\tilde{f}_i(t_j)I(Q_i \ge t_j)}
{1-\tilde{p}_i\sum_{h:h<Q_i}\tilde{f}_i(t_h)}
	\bigg\}
	e^{\tilde{\bfbeta}\top\vecZ_i}. \label{eqn:tilde lam}
\end{align}
By the mean value theorem,
\begin{equation}\label{eqn:dS-fMVT}
0 \le e^{\lambda_k\phazi}-1-\lambda_k\phazi \le \frac{1}{2}\left(\lambda_k\phazi\right)^2 e^{\lambda_k\phazi} \le \frac{1}{2}\upM^2\lambda_k^2e^{\lambda_k\upM}.
\end{equation}
where $ \upM$ is defined  in \eqref{def:Mm}.
Applying \eqref{eqn:dS-fMVT} to the denominator in \eqref{eqn:tilde lam},  we get
\begin{equation*}
1-\tilde{p}_i\sum_{h:h<Q_i}\tilde{f}_i(t_h) \ge 1-\tilde{p}_i\{1-\tilde{S}_i(Q_i)\}.
\end{equation*}
 By a similar argument, we have almost surely
\bes
& & \tilde{S}_i(t_k)I(Q_i \ge t_k)-\sum_{j\ge k}\tilde{f}_i(t_j)I(Q_i \ge t_j) \\
&=& \tilde{S}_i(Q_i)I(Q_i \ge t_k)+\sum_{j\ge k}\left\{1-e^{-\tilde{\lambda}_j e^{\tilde{\bfbeta}\top\vecZ_i}}-\tilde{\lambda}_j e^{\tilde{\bfbeta}\top\vecZ_i}\right\}\tilde{S}_i(t_j)I(Q_i > t_j)\\
&\le & \tilde{S}_i(Q_i)I(Q_i \ge t_k).
\ees 
Then,  $\tilde{\bfth}$ satisfies \eqref{eqn:EMjump}. 

For $\hat{\bfth}$, it must satisfy the score equation for $\lambda_k$'s:
$$
\frac{\partial l(\bfth)}{\partial \lambda_k} = \sum_{i=1}^n \left\{ \frac{d N_i(t_k)}{\lambda_k}
-W^{\bfth}_i(t_k)\phazi\right\}= 0, \quad \forall k=1,\dots,K.
$$
This is the equation version of \eqref{eqn:EMjump} after rearrangement. 

Now let $\hat{\lambda}_{\max,\zeta}$ and $\tilde{\lambda}_{\max,\zeta}$ be the largest jump for $\hat{\Lambda}$ and $\tilde{\Lambda}$ on $[0,\zeta]$, correspondingly. By Lemma \ref{lem:bound baseline} part \ref{lempart:Whead},
we have 
$$\limsup_{n\to \infty}n\hat{\lambda}_{\max,\zeta} \le C_w^{-1}, \quad
\limsup_{n\to \infty}n\tilde{\lambda}_{\max,\zeta} \le C_w^{-1}, a.s..$$
Hence \eqref{eqn:Omega0_to_1} is established.

In the set $\Theta$, we evaluate the discrepancy between $\log \tilde{L}(\bfth)$ and $\log L (\bfth)$, which 
can be bounded as following
\begin{equation}
1-S_i(Q_i)-\sum_{k:t_k<Q_i}\lambda_k\phazi S_i(t_k)
=\sum_{k:t_k<Q_i} S_i(t_k) \left(e^{\lambda_k\phazi}-1-\lambda_k\phazi\right).
\end{equation}
Applying \eqref{eqn:dS-fMVT} to $|\log L(\bfth)-\log\tilde{L}(\bfth)|$, 
we have the bound
\begin{align*}
\left|\log L(\bfth)-\log\tilde{L}(\bfth)\right|\le& \sum_{i=1}^n\left|
		\log\left\{ 1-p_i+p_iS_i(Q_i)\right\} -\log\left\{1-p_i\sum_{k:t_k<Q_i}\lambda_k\phazi S_i(t_k)\right\}
		\right| \notag \\
	\le & \sum_{i=1}^n\left|\frac{p_i}{1-p_i}\frac{n}{2}\upM^2\lambda_k^2e^{\lambda_k\upM}\right|
	\le  \frac{1}{2}n^2 e^{\upM\lambda_{\max,\zeta}} \upM^3\lambda_{\max,\zeta}^2.
\end{align*}
Using the upper bound for $\lambda_{\max,\zeta}$ in $\Theta$, we can bound
\begin{equation}\label{bound:L-L}
\sup_{\bfth \in \Theta}\left|\log L(\bfth)-\log\tilde{L}(\bfth)\right| \le
	e^{\frac{2\upM}{C_w}}\frac{2\upM^3}{C_w^2}.
\end{equation}

In summary whenever $\hat{\bfth}, \tilde{\bfth}\in {\Theta}$,
we have 
\begin{equation}\label{eqn:bound_on_Omega0}
0 \le \log L(\hat{\bfth}) - \log L(\tilde{\bfth}) \le \log L(\hat{\bfth})-\log\tilde{L}(\hat{\bfth})
	+ \log\tilde{L}(\tilde{\bfth}) - \log L(\tilde{\bfth}) <e^{\frac{2\upM}{C_w}}\frac{4\upM^3}{C_w^2} .
\end{equation}
Combining \eqref{eqn:bound_on_Omega0} and \eqref{eqn:Omega0_to_1} completes the proof.
\end{proof}

\begin{proof}[\underline{Proof of Theorem \ref{cor:EMconsistency} and \ref{cor:EMweakconv}}] 
From Lemma \ref{lemma:EM}, we only need  to establish the following two facts: 
1) $ 
\E [l_1(\bfth )]$
exists with one unique maximal, and 2) it is locally invertible at the maximal. 
We will see that
1) is verified through the proof of Theorem \ref{thm:consistency}, and  2) is
verified through the proof of Theorem  \ref{thm:normal}. 
\end{proof}

\subsection{Consistency of NPMLE}

\begin{proof}[\underline{Proof of Theorem \ref{thm:consistency}}]
The constants $\upM$, $c$, $\varepsilon$ and $\mathcal{L}$ are defined  in \eqref{def:Mm}, \eqref{def:vareps} and \eqref{def:Lips}. 

First, we show that the ``bridge'' $\bar{\Lambda}$ defined in \eqref{def:Lam bar}
converges to the true $\Lambda_0$ in the following sense: 
\begin{equation}\label{lempart:Lambar0}
 	 \sup_{t\in [0,\tau]}\left|e^{-\bar{\Lambda}(t)}-e^{-\Lambda_0(t)}\right|\to 0,  a.s.
\end{equation}
as $n\rightarrow\infty$. 
   We have the bound for $\forall t \in (0,\tau)$,
   \begin{equation}\label{bound:dLambar}
     \upM \ge \frac{\E \left[Y(t)\phi^{\bfth_0}(t)e^{\bfbeta_0^\top\vecZ}\right]}{\E \left[\log \left\{1+\exp\left(\bfalpha_0^\top \vecA-\Lambda_0(t)\phaz0\right)\right\}\right]} \ge \frac{\varepsilon}{\upM^2+\upM}.
    \end{equation}
   For any $\tau^*<\tau$ in $\Q$ the set of rational numbers, $\E [ Y(t)\phi^{\bfth_0}(t)e^{\bfbeta_0^\top\vecZ} ]$ is bounded away from zero over $[0, \tau^*]$. 
The uniform convergence of $\bar{\Lambda}$ to $\Lambda_0$ over any $[0, \tau^*]$ can be obtained in the way like \cite{Murphy94}. To extend the result to \eqref{lempart:Lambar0}, we use a trick described in \eqref{eq:tau*-tau1}-\eqref{eq:tau*-tau4}.
By Assumption \ref{assumpt:baseline}, $\Lambda_0$ is non-decreasing and diverges to $\infty$ at $\tau$.
Therefore,
 \begin{equation}\label{eq:tau*-tau1}
   \forall \epsilon>0, \, \exists \tau^* \in (0,\tau)\cap \Q, \, s.t. \, e^{-\Lambda_0(\tau^*)}<\epsilon/3.
 \end{equation}
 Through Rao's law of large number and Helly-Bray argument, we have
 \begin{equation}\label{eq:tau*-tau2}
   \sup_{t\in[0,\tau^*]}|\bar{\Lambda}(t)-\Lambda_0(t)| \to 0, \quad a.s. .
 \end{equation}
By continuity of the exponential function,
  \begin{equation}\label{eq:tau*-tau3}
  \exists N, \, \forall n>N, \, \sup_{t\in[0,\tau^*]}|e^{-\bar{\Lambda}(t)}-e^{-\Lambda_0(t)}|<\epsilon/3.
  \end{equation}
Then,
  \begin{equation}\label{eq:tau*-tau4}
  \forall n>N, \, \sup_{t\in[\tau^*,\tau]}|e^{-\bar{\Lambda}(t)}-e^{-\Lambda_0(t)}|
    \le 2e^{-\Lambda_0(\tau^*)}+|e^{-\bar{\Lambda}(\tau^*)}-e^{-\Lambda_0(\tau^*)}| <\epsilon.
  \end{equation}
Therefore, we have proved \eqref{lempart:Lambar0}. 

Next, we evaluate the difference between the limits of $\hat{\Lambda}$ and $\bar{\Lambda}$. 
According to Assumption \ref{assumpt:parameter} and $e^{-\hat{\Lambda}(t)} \in [0,1]$, 
$(\hat{\bfalpha},\hat{\bfbeta},e^{-\hat{\Lambda}(t)})$ is bounded. 
$\hat{\Lambda}(t)$ is C\`{a}dl\`{a}g, so is $e^{-\hat{\Lambda}(t)}$.
By Helly's Selection theorem, there is a subsequence converging uniformly almost surely to some
    $\bfth^*=(\bfalpha^*, \bfbeta^*, e^{-\Lambda^*})$.  
Lemma \ref{lem:bound baseline} part \ref{lempart:Whead} gives the bound for $\E \{ W^{\bfth}(t)e^{\bfbeta^\top\vecZ} \}$ over $[0,\zeta]$. 
We only need to find its bound on $[\zeta,\tau]$ in order to mimic the proof of Lemma 1 of \cite{Murphy94}. 
Note that 
\begin{align*}
     \E \left[ W^{\bfth}(t)e^{\bfbeta^\top\vecZ} \right]
        =&\E \left[\int_{t-}^\tau \big\{1-\phi^{\bfth}(u)\big\}e^{\bfbeta^\top\vecZ}dN(u) \right] \notag \\
        & -\E \left[\int_t^\tau \phi^{\bfth}(u)e^{\bfbeta^\top\vecZ}d\E[Y(u)|\vecA,\vecZ] \right].
   \end{align*}
By Assumption \ref{assumpt:ltrc}, $\P(Q_i \le\zeta)=1$, so $\E[Y(u)|\vecA,\vecZ]$
is decreasing on $[\zeta,\tau]$. Along with the Lipschitz continuity,  we have for $\forall t \in [\zeta,\tau)$
\begin{equation*}
     M\mathcal{L} \ge \frac{\E [W^{\bfth}(t)e^{\bfbeta^\top\vecZ}]}{\E \left[\log \left\{1+\exp\left(\bfalpha_0^\top \vecA-\Lambda_0(t)\phaz0\right)\right\}\right]} \ge \frac{\varepsilon}{\upM^2+\upM}.
    \end{equation*}
Therefore, $\gamma(t)=\frac{\E \left[ W^{\bfth_0}(t)e^{\bfbeta^\top\vecZ}\right]}
{\E \left[ W^{\bfth^*}(t)e^{\bfbeta^{*\top}\vecZ}\right]}
    $ is bounded away from both $\infty$ and zero,
and 
\begin{equation}\label{lempart:ratio baseline}
\sup_{t \in [0,\tau]}\left| \frac{d\hat{\Lambda}}{d\bar{\Lambda}}(t)-\gamma(t) \right| \to 0
\text{ and}
\sup_{t \in [0,\tau^*]}\left| \hat{\Lambda}(t)-\int_0^t\gamma d\Lambda_0 \right| \to 0 \; a.s.
, \forall \tau^*<\tau \text{ in }\Q. 
\end{equation} 
 
After all these preparation, we can use the semi-parametric Kullback-Leibler divergence argument from \cite{Murphy94}. 
We have 
 \begin{align}
0  \le  & \frac{1}{n} \big\{ l_n(\hat{\bfalpha},\hat{\bfbeta},\hat{\Lambda})
            -l_n(\bfalpha_0,\bfbeta_0,\bar{\Lambda}) \big\} \notag\\
\notag    =& \frac{1}{n}\sum_{i=1}^n
   \int_0^\tau \log\bigg\{ \frac{\phi_i^{\hat{\bfth}}(u)e^{\hat{\bfbeta}^\top \vecZ_i}d \hat{\Lambda}(u) }
   {\phi_i^{\bfth_0}(u)e^{\bfbeta_0^\top \vecZ_i}d \bar{\Lambda}(u)} \bigg\} \bigg\{ dN_i(u)- \phi_i^{\bfth_0}(u) Y_i(u)  e^{\bfbeta_0^\top \vecZ_i} d\bar{\Lambda}(u)\bigg\}\\
  \notag & + \int_0^\tau \left[ \log\bigg\{ \frac{\phi_i^{\hat{\bfth}}(u)e^{\hat{\bfbeta}^\top \vecZ_i}d \hat{\Lambda}(u) }
   {\phi_i^{\bfth_0}(u)e^{\bfbeta_0^\top \vecZ_i}d \bar{\Lambda}(u)} \bigg\}
    - \bigg\{ \frac{\phi_i^{\hat{\bfth}}(u)e^{\hat{\bfbeta}^\top \vecZ_i}d\hat{\Lambda}(u)}
   {\phi_i^{\bfth_0}(u)e^{\bfbeta_0^\top \vecZ_i}d\bar{\Lambda}(u)}-1\bigg\} \right] \\
   & \qquad \times \phi_i^{\bfth_0}(u)e^{\bfbeta_0^\top \vecZ_i}Y_i(u) d\bar{\Lambda}(u).
\label{eq:KL}
\end{align}
Denote the function in the logarithm above as $\psi_i(u)$. Using the definition of $\bar{\Lambda}$, we can rewrite the first term in \eqref{eq:KL} as
\begin{align}
 &\frac{1}{n}\sum_{i=1}^n \left\{
   \int_0^\tau \log\big(\psi_i(u)\big) - \frac
   {\sum_{j=1}^n \log\big(\psi_j(u)\big)\phi_j^{\bfth_0}(u)Y_j(u)e^{\bfbeta_0^\top \vecZ_j}}
{\sum_{j=1}^n \phi_j^{\bfth_0}(u)Y_j(u)e^{\bfbeta_0^\top \vecZ_j}} \right\}
 dN_i(u)  \notag \\
= & \frac{1}{n}\sum_{i=1}^n \left\{
   \int_0^\tau \log\big(\psi_i(u)\big) - \frac
   {\sum_{j=1}^n \log\big(\psi_j(u)\big)\phi_j^{\bfth_0}(u)Y_j(u)e^{\bfbeta_0^\top \vecZ_j}}
{\sum_{j=1}^n \phi_j^{\bfth_0}(u)Y_j(u)e^{\bfbeta_0^\top \vecZ_j}} \right\}
 dM_i(u) \label{term:KL1}
\end{align}
Inside $\psi_i(u)$, the ratio $d\hat{\Lambda}/d\bar{\Lambda}$ is bounded away from $0$ and $\infty$ according to \eqref{lempart:ratio baseline}. 
Denote the range of the ratio as $[1/R,R]$. The $\phi_i^{\bfth_0}(u)$ term and  $\phi_i^{\hat{\bfth}}(u)$ term in $\psi_i(u)$ creates potential singularity for \eqref{term:KL1} at $\tau$, 
but its decay rate is bounded by $e^{-\upM R \Lambda_0(u)}$ by Assumptions \ref{assumpt:parameter} and \ref{assumpt:cov bound}. The integrands of martingale integral \eqref{term:KL1} are all bounded a.s., 
and the quadratic variation of \eqref{term:KL1} is bounded a.s. by
\begin{equation*}
\frac{1}{n^2}\sum_{i=1}^n
   \int_0^\tau 4\big\{\upM R \Lambda_0(u) + \log(R) \big\}^2 \phi_i^{\bfth_0}(u) Y_i(u)e^{\bfbeta_0^\top \vecZ_i} d\Lambda_0(u). 
\end{equation*}
It is of order $O_p(1/n)$, so the limit of \eqref{term:KL1} is zero almost surely. 

The integrands in the second term of \eqref{eq:KL} is of the form $\log(x)-(x-1) \le 0$. 
In order to satisfy the inequality in \eqref{eq:KL}, 
we must have
\begin{equation*}
\lim_{n \to \infty}  \frac{1}{n}\sum_{i=1}^n  \int_0^\tau \big\{\log\big(\psi_i(u)\big)  - \big(\psi_i(u) -1 \big)\big\}
\phi_i^{\bfth_0}(u)e^{\hat{\bfbeta}^\top \vecZ_i}Y_i(u) d\bar{\Lambda}(u)= 0. 
\end{equation*}
Applying the same argument as in \cite{Murphy94}, we get
\begin{equation}\label{eq:identifiability}
   \E\left(\int_0^\tau \left| \phi^{\bfth^*}(u)e^{\bfbeta^{*\top} \vecZ}\gamma(u)
- \phi^{\bfth_0}(u)e^{\bfbeta_0^\top \vecZ} \right|Y(u)d\Lambda_0(u)\right)=0
 \end{equation}
in the almost sure set. 
The identifiability of our model is verified in \cite{Li01} Theorem 2. 
Along with our regularity conditions in Assumptions \ref{assumpt:cov bound} and \ref{assumpt:baseline}, 
\eqref{eq:identifiability} leads to $\bfalpha^*=\bfalpha_0$, $\bfbeta^*=\bfbeta_0$ and $\gamma(t)=1$.  
This implies that 
$$\sup_{t \in [0,\tau^*]}\left| \hat{\Lambda}(t)-\Lambda_0 (t)\right| \to 0 \; a.s.
, \forall \tau^*<\tau \text{ in }\Q. $$
Repeating the trick in \eqref{eq:tau*-tau1}-\eqref{eq:tau*-tau4}, we have
$$
\sup_{t \in [0,\tau]}\left| e^{-\hat{\Lambda}(t)}-e^{-\Lambda_0 (t)}\right| \to 0 \; a.s..
$$

Finally, we summarize all usage of almost sure arguments to ensure that intersection of all almost sure sets  still has probability one under $\sigma$-additivity. 
 The steps \eqref{eq:tau*-tau1}-\eqref{eq:tau*-tau4} involves one almost sure argument for each choice of $\tau^*$. We preserve the almost sure property by restricting $\tau^*$ to be in the countable set $\Q$. One almost sure argument is made for Helly's selection theorem. In Lemma \ref{lem:bound baseline}, we use the Glivenko-Cantelli Theorem to avoid the dependence on the choice of $\bfth^*$, so the almost sure argument is only applied once. Two more almost sure arguments are used in calculating the limit of the terms in \eqref{eq:KL}.

\end{proof}

\begin{proof}[\underline{Proof of Theorem \ref{thm:consistency'}}]
The proof is essentially the same as the proof of Theorem \ref{thm:consistency}, 
so the details are omitted. 
In fact, it is less technical due to the boundedness of $\Lambda_0$ over $[0, \tau']$. 
\end{proof}

  \subsection{Asymptotic Normality}

First, we provide the definition of several quantities below. 
In Theorem \ref{thm:normal} $\sigma(\bfh)=\Big(\bfsig_a(\bfh),\bfsig_b(\bfh),\sigma_\hbaseline(\bfh)\Big)$  is
  \begin{align}
 \notag   \bfsig_a(\bfh)= \E\Bigg[& \vecA
    \bigg\{ -\int_0^{\tau'}  K^{\bfth_0}_1(\bfh)(u)Y(u)d\phi^{\bfth_0}(u) \\
    & + K^{\bfth_0}_2(\bfh)Y(\tau')\phi^{\bfth_0}(\tau')\Big(1-\phi^{\bfth_0}(\tau')\Big)
\bigg\}\Bigg], \notag \\
\notag    \bfsig_b(\bfh)= \E\Bigg[& \vecZ
  \bigg\{\int_0^{\tau'}  K^{\bfth_0}_1(\bfh)(u)Y(u)\phazE d\Big[\Lambda_0(u)\phi^{\bfth_0}(u)\Big]\\
  &- K^{\bfth_0}_2(\bfh)Y(\tau')\phazE\Lambda_0(\tau')\phi^{\bfth_0}(\tau')\Big(1-\phi^{\bfth_0}(\tau')\Big)\bigg\}\Bigg], \notag \\
\notag    \sigma_\hbaseline(\bfh)=\E\Bigg[& \phazE\bigg\{  K_1^{\bfth_0}(\bfh)(u)\phi^{\bfth_0}(u)Y(u)
            - K_2^{\bfth_0}(\bfh)Y(\tau') \phi^{\bfth_0}(\tau')\Big(1-\phi^{\bfth_0}(\tau')\Big) \\
  &-\int_u^{\tau'} K^{\bfth_0}_1(\bfh)(s)\phi^{\bfth_0}(s)\Big(1-\phi^{\bfth_0}(s)\Big)Y(s)\phazE d\Lambda_0(s)\bigg\}\Bigg], \label{def:sig}
  \end{align}
where
\begin{align}
\notag K_1^{\bfth}(\bfh)(u)=&
    \mathbf{a}^\top \vecA \Big(1-\phi^{\bfth}(u)\Big)
    +\mathbf{b}^\top \vecZ
    \left\{1-\Big(1-\phi^{\bfth}(u)\Big)
    \Lambda(u)e^{\bfbeta^\top \vecZ}\right\}\\
     &+\hbaseline(u)-\Big(1-\phi^{\bfth}(u)\Big)
    e^{\bfbeta^\top\vecZ}  \int_0^u \hbaseline d\Lambda, \notag \\
K_2^{\bfth}(\bfh)=&\Big\{\bfa^\top\vecA-\bfb^\top\vecZ\Lambda(\tau') e^{\bfbeta^\top\vecZ} -\int_0^{\tau'}\hbaseline  e^{\bfbeta^\top\vecZ} d\Lambda \Big\}. \label{def:K}
\end{align}

Let 
$
\bfth+t\bfh=\Big(\bfalpha+t\bfa,\bfbeta+t\bfb,\int_0^\cdot (1+t\hbaseline )d\Lambda\Big)
$.
Define the directional derivatives
\begin{equation*}
 \lim_{t\to 0}\frac{l^I_n(\bfth+t\bfh)-l^I_n(\bfth)}{t}
 =S^{\bfth}_n=S^{\bfth}_{n,a}+S^{\bfth}_{n,b}+S^{\bfth}_{n,\hbaseline},
\end{equation*}
where
\begin{align*}
\notag S^{\bfth}_{n,a}= &\frac{1}{n}\sum_{i=1}^n \mathbf{a}^\top \vecA_i
    \bigg\{
    \int_0^{\tau'} \Big(1-\phi_i^{\bfth}(u)\Big) dN_i(u)
    -\int_0^{\tau'} Y_i(u)\phi_i^{\bfth}(u)\Big(1-\phi_i^{\bfth}(u)\Big)
    e^{\bfbeta^\top \vecZ_i}d\Lambda(u)\\
\notag & +\Big(N_i(\tau)-N_i(\tau')\Big)\Big(1-\phi_i^{\bfth}(\tau')\Big)
        -Y_i(\tau)\phi_i^{\bfth}(\tau')\bigg\}\\
\notag S_{n,b}^{\bfth}=&\frac{1}{n}\sum_{i=1}^n \mathbf{b}^\top \vecZ_i
    \bigg[
    \int_0^{\tau'} \left\{1-\Big(1-\phi_i^{\bfth}(u)\Big)
    \Lambda(u)e^{\bfbeta^\top \vecZ_i}\right\} dN_i(u)\\
 \notag   &+\int_0^{\tau'} Y_i(u)\phi_i^{\bfth}(u)e^{\bfbeta^\top \vecZ_i}
        \left\{ \Big(1-\phi_i^{\bfth}(u)\Big)\Lambda(u)e^{\bfbeta^\top \vecZ_i}
        -1
        \right\}
    d\Lambda(u)
\\
    \notag & -\Big(N_i(\tau)-N_i(\tau')\Big)\Big(1-\phi_i^{\bfth}(\tau')\Big)\Lambda(\tau')\phazi
        +Y_i(\tau)\phi_i^{\bfth}(\tau')\Lambda(\tau')\phazi\bigg]\\
\notag S_{n,\hbaseline}^{\bfth} =& \frac{1}{n}\sum_{i=1}^n \int_0^{\tau'} \left[\hbaseline(u)-\Big\{1-\phi_i^{\bfth}(u)\Big\}
    e^{\bfbeta^\top\vecZ_i}  \int_0^u \hbaseline d\Lambda\right] dN_i(u) \\
\notag    &+ \int_0^{\tau'} Y_i(u)\phi_i^{\bfth}(u)e^{\bfbeta^\top \vecZ_i}
    \left[
    \Big\{1-\phi_i^{\bfth}(u)\Big\}
    e^{\bfbeta^\top\vecZ_i}  \int_0^u \hbaseline d\Lambda
    -\hbaseline(u)
    \right]
    d\Lambda(u)\\
    \notag & -\Big(N_i(\tau)-N_i(\tau')\Big)\Big(1-\phi_i^{\bfth}(\tau')\Big)\int^{\tau'}_0\hbaseline d\Lambda\phazi
        +Y_i(\tau)\phi_i^{\bfth}(\tau')\int^{\tau'}_0\hbaseline d\Lambda\phazi .
\end{align*}
Their expectations are denoted as
\begin{equation*}
  S^{\bfth}=S^{\bfth}_a+S^{\bfth}_b+S^{\bfth}_\hbaseline=\E \left(S^{\bfth}_{n,a}\right)+\E \left( S^{\bfth}_{n,b}\right)+\E \left( S^{\bfth}_{n,\hbaseline}\right) .
\end{equation*}

 Again let $\bfth_0$ be the true parameter and $\bfth$ another element in the paramter space. Define $\triangle \bfth=\bfth-\bfth_0$ with
\begin{equation*}  
\triangle \bfalpha=\bfalpha-\bfalpha_0, \,
\triangle \bfbeta=\bfbeta-\bfbeta_0 \text{ and } 
\triangle \Lambda(\cdot)=\Big\{\Lambda(\cdot)-\Lambda_0(\cdot)\Big\}.
\end{equation*}
Define $lin \Theta $ to be the linear space spanned by $\{ \bfth -\bfth_0 : \bfth  \text{ in parameter space}\}$. 
Let $\bfth_t = \bfth_0+t\triangle \bfth$.  
The functional Hessian is a linear operator $lin \Theta \mapsto l^\infty(H_p)$ defined as
\begin{align}
\dot{S}^{\bfth_0}(\triangle \bfth)(\bfh) = & 
		\lim_{t\to 0}\frac{S^{\bfth_t }(\bfh)-S^{\bfth_0}(\bfh)}{t} \notag \\
= & -\triangle\bfalpha^\top \bfsig_a(\bfh)
 -\triangle\bfbeta^\top \bfsig_b (\bfh)
-\int_0^{\tau'} \sigma_\hbaseline(\bfh)(u)d\triangle\Lambda(u) \label{def:funcHess}
\end{align}
with $\sigma$ defined in \eqref{def:sig}. 

The following Lemma \ref{lem:sig inv} is used in the proofs of Theorems \ref{thm:normal} and \ref{thm:var est}. It tells us about the property of $\sigma$, the essential element in the functional Hessian. 

\begin{lemma}\label{lem:sig inv}
    Let the operator $\sigma: (\bfa,\bfb,\hbaseline) \mapsto \Big(\bfsig_a(\bfh),\bfsig_b(\bfh),\sigma_\hbaseline(\bfh)\Big)$
    be defined as in \eqref{def:sig}.
  Under the conditions of Theorem \ref{thm:normal},
  $\sigma$ is a continuously invertible bijection from $H_\infty$ to $H_\infty$.
 \end{lemma}
\begin{proof}[\underline{Proof of Lemma \ref{lem:sig inv}}]
  First we prove that $\sigma$ is injection by an identifiability argument.
  Define an inner-product between $\sigma(\bfh)$ and $\bfh$ as 
\begin{align*}
 \Big<\sigma(\bfh),\bfh\Big>=& \bfa^\top\bfsig_a(\bfh)+\bfb^\top\bfsig_b(\bfh)+\int_0^{\tau'}\sigma_\hbaseline(\bfh)(u)\hbaseline(u)
    d\Lambda_0(u) \\
 \notag   =& \int_0^{\tau'}\E  \left[\big\{K^{\bfth_0}_1(\bfh)(u)\big\}^2Y(u)\phi^{\bfth_0}(u)\phazE\right] d\Lambda_0(u)\\
&+\E \left[ \big\{ K^{\bfth_0}_2(\bfh)\big\}^2Y(\tau')\phi^{\bfth_0}(\tau')\Big(1-\phi^{\bfth_0}(\tau')\Big) \right].
\notag
 \end{align*}
 If $\Big<\sigma(\bfh),\bfh\Big>=0$, we have almost surely $K^{\bfth_0}_2(\bfh)=0$ and $K^{\bfth_0}_1(\bfh)(u)=0$ a.e. $u \in [0, \tau']$.  
 Therefore,
$$
\int_0^t  K^{\bfth_0}_1(\bfh)(u)\phi^{\bfth_0}(u)\phazE d\Lambda_0(u)=0, \forall t\in[0,\tau'], a.s..
$$
Calculating the integral, we have for for any $t\in[0,\tau']$ a.s.
$$
-\bfa^\top\vecA\phi^{\bfth_0}(t)+\bfb^\top\vecZ\phi^{\bfth_0}(t)\Lambda_0(t)\phazE +\int_0^t\hbaseline(u)d\Lambda_0(u)\phi^{\bfth_0}(t)\phazE=0.
$$
Setting $t=0$, we have $-\bfa^\top\vecA\phi^{\bfth_0}(0)=0$, so $\bfa^\top\vecA=0$.
By Assumption \ref{assumpt:cov bound}, 
$\bfa=0$.
Plugging $\bfa=0$ into $K^{\bfth_0}_2$ yields
$$
 K^{\bfth_0}_2(\bfh)= \phazE\Big\{\bfb^\top\vecZ\Lambda_0(\tau')-\int_0^{\tau'}\hbaseline(u)d\Lambda_0(u)\Big\}=0, a.s..
$$
Again, $\bfb^\top\vecZ = \int_0^{\tau'}\hbaseline(u)d\Lambda_0(u)/\Lambda_0(\tau')$ is deterministic, so $\bfb=0$. This way $\hbaseline$ must also be constantly zero.
As a result, $\sigma(\bfh)=\sigma(\bfh') \Rightarrow \Big(\sigma(\bfh-\bfh'),\bfh-\bfh'\Big)=0 \Rightarrow \bfh=\bfh'$.

To show it is a bijection, we apply Theorem 3.11 in \cite{Conway97}.
It suffices to decompose $\sigma$ as the sum of one invertible operator
and one compact operator.
The invertible operator is defined as
\begin{equation*}
\Sigma(\bfh)=\Big(\E \left(\vecA\vecA^\top\right) \bfa,\E \left(\vecZ\vecZ^\top\right) \bfb , \hbaseline(t)\E \left\{\phazE \phi^{\bfth_0}(t)Y(t)\right\}\Big).
\end{equation*}
Since $\E \left(\vecA\vecA^\top\right)$, $\E \left(\vecZ\vecZ^\top\right)$ are both positive definite, and $\inf_{t\in[0,\tau']}\E \phazE \phi^{\bfth_0}(t)Y(t)>0$, the inverse exists as
\begin{equation*}
\Sigma^{-1}(\bfh)=\Big(\left[\E \big\{\vecA\vecA^\top \big\}\right]^{-1}\bfa,\left[\E \big\{\vecZ\vecZ^\top\big\}\right]^{-1} \bfb , \hbaseline(t)\left[\E \big\{\phazE \phi^{\bfth_0}(t)Y(t)\big\}\right]^{-1}\Big).
\end{equation*}
For the compactness of $\sigma(\bfh)-\Sigma(\bfh)$, classical Helly-selection plus dominated convergence method applies as all terms are conveniently bounded.
\end{proof}

The proof of Theorem \ref{thm:normal} is the application of Theorem 3.3.1 from \cite{vandervaart1996}. 
We shall verify all the required conditions for the Theorem.

 \begin{proof}[\underline{Proof of Theorem \ref{thm:normal}}]

Since we work under a modified Assumption \ref{assumpt:baseline'} now, the martingale representation in \eqref{def:mart1} needs to change accordingly beyond $\tau'$. 
We still use $M_i(t)$ as the notation. 
Define the filtrations $\big\{\mathcal{F}_t: t \in [0,\tau] \big\}$. 
On $[0,\tau']$, $\mathcal{F}_t$ is the natural $\sigma$-algebra generated by
$\{N_i(t), Y_i(t), \vecA_i, \vecZ_i, i = 1, \dots, n\}$. 
Since there is no extra information in the tail window $(\tau', \tau)$, 
we set $\mathcal{F}_t =\mathcal{F}_{\tau'}$ for $t \in (\tau', \tau)$. 
$\mathcal{F}_\tau$ is the $\sigma$-algebra generated by
$\{N_i(\tau)-N_i(\tau'), Y_i(\tau), \vecA_i, \vecZ_i, i = 1, \dots, n\}$,
where $Y_i(\tau) = Y_i(\tau') - dN_i(\tau')$ is measurable in $\mathcal{F}_{\tau'}$. 
The filtrations on $[0,\tau']$ stay the same, so $M_i(t)$ defined in \eqref{def:mart1} is still a martingale up to time $\tau'$. 
In the tail window $(\tau', \tau)$, we set $M_i(t)$ constantly equals $M_i(\tau')$. 
To extend its definition to time $\tau$, we define
\begin{equation}\label{def:mart2}
 d M_i(\tau) = M_i(\tau) - M_i(\tau') = \big\{N_i(\tau)-N_i(\tau')\big\} -  Y_i(\tau) \phi^{\bfth_0}_i(\tau').
\end{equation}
It is easy to verify that $\E[ M_i(\tau)| \mathcal{F}_{\tau'}] = M_i(\tau')$, 
so $M_i(t)$ thus defined is a martingale with respect to the new filtrations $\big\{\mathcal{F}_t: t \in [0,\tau'] \cup \{\tau\}\big\}$. 
Analogously, we define the process $M^{\bfth}_i(\cdot)$ which replaces the true parameter $\bfth_0$ in $M_i(\cdot)$ by arbitrary $\bfth$ in the parameter space. 
Apparently, $M^{\bfth_0}_i(\cdot) = M_i(\cdot)$. 
From here, we establish the needed results based on the martingale theory. 

First, we prove  weak convergence of the empirical score
 \begin{equation}\label{lem:weak conv}
   \sqrt{n}(S^{\bfth_0}_n-S^{\bfth_0})\stackrel{l^\infty(H_p)}{\longrightarrow} \mathcal{W}.
 \end{equation}
Notice that
$S_1^{\bfth_0}-S^{\bfth_0}$ is a martingale integral with respect to \eqref{def:mart2}.
The weak convergence follows from martingale central limit theorem.
The covariance process is given by the expectation of its quadratic variation:
 \begin{align*}
\notag  \Cov\big(\mathscr{G}(\bfh),\mathscr{G}(\bfh^*)\big)=\E\Big[ \int_0^{\tau'}&
     K^{\bfth_0}_1(\bfh) K^{\bfth_0}_1(\bfh^*)
    Y(u)\phi_0(u)e^{\bfbeta_0^\top \vecZ}d\Lambda_0(u) \\
    &+ K^{\bfth_0}_2(\bfh) K^{\bfth_0}_2(\bfh^*)\phi_0(\tau')\big\{1-\phi_0(\tau')\big\}
	\Big],
\end{align*}
where $K_1$ and $K_2$ are defined as in \eqref{def:K}.

Next, we verify the approximation condition 
\begin{equation}\label{lem:approximation}
   \sqrt{n}\left(S_n^{\hat{\bfth}}-S^{\hat{\bfth}}
- S_n^{\bfth_0}-S^{\bfth_0}\right) = o_p(1).
 \end{equation}
Consider the class $\{S_1^{\bfth}(\bfh)-S_1^{\bfth_0}(\bfh): \|\bfth-\bfth_0\| \le \varepsilon, \bfh\in H_p \}$. 
All terms involved in this class are uniformly bounded with uniformly bounded variation, so it is a Donsker class for the set of observable random variables.
By checking that $\phi_i^{\bfth}$ is Lipschitz in $\bfth$ under the $l^\infty(H_p)$ norm, 
we have almost surely
$$
\sup_{t,\vecZ,\vecA} |\phi_i^{\bfth}(t)-\phi_i^{\bfth_0}(t)|
= O\left(\|\bfth-\bfth_0\|\right),
$$
and similarly
\begin{equation*}
\sup_{t,\vecZ,\vecA} |\phi_i^{\bfth}(t) {\Lambda}(t)-\phi_i^{\bfth_0}(t)\Lambda_0(t)|
=O\left(\|\bfth-\bfth_0\|\right).
\end{equation*}
For a single summand in the score, 
\begin{equation*}
  \sup_{h\in H_p}\E[S_1^{\bfth}(\bfh)-S_1^{\bfth_0}(\bfh)]^2
=O\left(\|\bfth-\bfth_0\|^2\right).
\end{equation*}
We plug $\hat{\bfth}$ into the expression above. Thus, the variance of the limiting process of \eqref{lem:approximation} is $o(1)$ by the consistency of $\hat{\bfth}$
from Theorem \ref{thm:consistency'}, so the process itself is $o_p(1)$. 

We then show the Fr\'{e}chet differentiability of expected score $S$ at $\bfth_0$ in the direction of $\hat{\bfth}-\bfth_0$,
\begin{equation}\label{lem:Frechet}
 S^{\hat{\bfth}_t}-S^{\bfth_0}=t\dot{S}^{\bfth_0}(\hat{\bfth}-\bfth_0)+o_p(t\|\hat{\bfth}-\bfth_0\|).
 \end{equation}
  We use a shorthand notation for the expected score at $\bfth$:
  $$
    S^{\bfth}(\bfh)= \E \left[\int_0^{\tau'} K_1^{\bfth}(\bfh)(u)dM^{\bfth}(u)
		+ 	K_2^{\bfth}(\bfh) d M^{\bfth}(\tau)\right] 
               =  \E \left[\int_0^{\tau} V^{\bfth}(\bfh)(u) d M^{\bfth}(u)\right],
  $$
by setting $$V^{\bfth}(\bfh)(t) = I(t \le \tau')K_1^{\bfth}(\bfh)(t) +  I(t=\tau) K_2^{\bfth}(\bfh).$$
  By the Lipschitz continuity  with respect to $\|\bfth\|$ for all terms involved, $ K_1^{\bfth}(\bfh)$, $K_2^{\bfth}(\bfh)$ and $dM^{\bfth}$,
  \bes 
&&\notag S^{ {\bfth}_t}(\bfh)-S^{\bfth}(\bfh) \\
&=& \E \left[\int_0^{\tau'} V^{ {\bfth}_t}(\bfh)(u)dM^{ {\bfth}_t}(u) \right] \\
\notag    &=& \E \left[ \int_0^{\tau'} V^{\bfth_0}(\bfh)(u)d\big\{M^{ {\bfth}_t}(u)- M^{\bfth_0}(u)\big\}\right]+\E \left[ \int_0^{\tau'}V^{ {\bfth}_t}(\bfh)(u)dM^{\bfth_0}(u)\right]\\
 &   &
    + \E \left[\int_0^{\tau'}\big\{V^{ {\bfth}_t}(\bfh)(u)- V^{\bfth_0}(\bfh)(u)\big\}d\big\{M^{ {\bfth}_t}(u)- M^{\bfth_0}(u)\big\} \right]\\
    &=& t\dot{S}^{\bfth_0}( {\bfth}-\bfth_0)(\bfh)+0+O_p(t^2\| {\bfth}-\bfth_0\|^2).
  \ees 
Again, we plug-in $\hat{\bfth}$ and use the consistency result to verify the condition \eqref{lem:Frechet}.

Afterwards, we find the local inverse of the functional Hessian in \eqref{def:funcHess}. 
We have shown in Lemma \ref{lem:sig inv} that the functional operator $\sigma$ is a continuously invertible bijection from $H_\infty$ to $H_\infty$. The invertibility of $\dot{S}^{\bfth_0}$ in $H_p$ follows from the following argument. 
By the continuous invertibility of $\sigma$, there is some $q$ so that $\sigma^{-1}(H_q) \subseteq H_p$, and
\begin{align}
  & \inf_{\triangle\bfth \in lin \Theta} \frac
  {\sup_{\bfh\in H_p}|(\bfalpha-\bfalpha_0)^\top \bfsig_a(\bfh)
        +(\bfbeta-\bfbeta_0)^\top \bfsig_b (\bfh)
        +\int_0^{\tau'}\sigma_\hbaseline(\bfh) d(\Lambda-\Lambda_0)|}
  { \|\triangle \bfth\|_{l^\infty(H_p)}\ } \notag  \\
  \ge & \inf_{\triangle\bfth \in lin \Theta} \frac
  {\sup_{\bfh\in \sigma^{-1}(H_q)}|(\bfalpha-\bfalpha_0)^\top \bfsig_a(\bfh)
        +(\bfbeta-\bfbeta_0)^\top \bfsig_b (\bfh)
        +\int_0^{\tau'}\sigma_\hbaseline(\bfh) d(\Lambda-\Lambda_0)|}
  {  p\|\triangle \bfth \| } \notag\\
=&\inf_{\triangle\bfth \in lin \Theta} \frac
  {\sup_{\bfh\in H_q}| \triangle \bfth (\bfh) |}
  { p\|\triangle \bfth \|  } >
  \frac{q }{2p}. \label{cor:cont inv}
\end{align}

Finally, let us put everything together. The NPMLE $\hat{\bfth}$ is shown to be consistent in Theorem \ref{thm:consistency'}, and 
  \eqref{lem:weak conv}, \eqref{lem:approximation}, \eqref{lem:Frechet} 
    and \eqref{cor:cont inv} verify the conditions of Theorem 3.3.1 from \cite{vandervaart1996}. 
 \end{proof}

  \begin{proof}[\underline{Proof of Theorem \ref{thm:var est}}]
  The proof for the continuous invertibility of $\hat{\sigma}$ is similar to the proof of  Lemma \ref{lem:sig inv}. The approximation error between the natural estimator $\hat{\sigma}$ 
  and Louis' formula variance estimator using \eqref{eq:louis} again comes from the ``ghost copies'' like the case in Lemma \ref{lemma:EM}, 
so the same argument applies to show their asymptotic equivalence. 
 \end{proof}


\setcounter{equation}{0}
\setcounter{lemma}{0}
\section{Appendix B: Variance Estimator}\label{appendixB}

\subsection{Derivatives of Log-likelihood}

Let $l^c(\bfalpha,\bfbeta,\blam)=\sum_{i=1}^n l^c_i(\bfalpha,\bfbeta,\blam)$ be the complete data log-likelihood,
\begin{align*}
\notag l^c_i(\bfalpha,\bfbeta,\blam)
 =& (\etaA_i+M_i) \bfalpha^\top \vecA_i -(1+M_i)\log(1+e^{\bfalpha^\top \vecA_i})\\
\notag &+\delta^1_i \etaA_i \sum_{k=1}^K I\{X_i=t_k\}(\log \lambda_k +\bfbeta^\top \vecZ_i)
- \etaA_i \sum_{k:t_k \le X_i} \lambda_k e^{\bfbeta^\top \vecZ_i}\\
 &+M_i\sum_{k:t_k<Q_i} I\{\kappa_i=k\}\Big(\log \lambda_k+\bfbeta^\top \vecZ_i-\sum_{h=1}^k \lambda_h e^{\bfbeta^\top \vecZ_i}\Big).
\end{align*}
Its gradient is given by
\begin{equation*}
    \nabla l^c_i=\left(\frac{\partial l^c_i}{\partial \bfalpha},
        \frac{\partial l^c_i}{\partial \bfbeta},
        \frac{\partial l^c_i}{\partial \blam}\right)^\top,
\end{equation*}
 where
\begin{align*}
  \frac{\partial l^c_i}{\partial \bfalpha}
=& \vecA_i\Big\{\etaA_i+M_i-(1+M_i)\frac{e^{\bfalpha^\top \vecA_i}}{1+e^{\bfalpha^\top \vecA_i}}\Big\}
= \vecA_i\big\{\etaA_i-p_i+M_i(1-p_i)\big\}, \\
\notag  \frac{\partial l^c_i}{\partial \bfbeta}
=& \vecZ_i \bigg\{\etaA_i \delta^1_i +M_i-\Big(\etaA_i \sum_{k:t_k \le X_i}\lambda_k
+M_i \sum_{k=1}^{\kappa_i} \lambda_k \Big)e^{\bfbeta^\top \vecZ_i}\bigg\} \\
 =& \vecZ_i \Big\{\etaA_i \delta^1_i +M_i-\etaA_i \Lambda_i(X_i)
-M_i \Lambda_i(\kappa_i)\Big\}, \\
\notag  \frac{\partial l^c_i}{\partial \lambda_k}
=& \Big(\etaA_i \delta^1_i  I\{X_i=t_k\}+M_i
I\{\kappa_i=k\}\Big)\frac{1}{\lambda_k}-\Big(\etaA_i I\{t_k \le X_i\}+M_iI\{\kappa_i \ge t_k\}\Big) e^{\bfbeta^\top \vecZ_i}\\
 =& \etaA_i\Big( \frac{\delta^1_i  I\{X_i=t_k\}}{\lambda_k}-I\{t_k \le X_i\}e^{\bfbeta^\top \vecZ_i}\Big)
+ M_i\Big( \frac{I\{\kappa_i=k\}}{\lambda_k}- I\{\kappa_i \ge t_k\} e^{\bfbeta^\top \vecZ_i}\Big).
\end{align*}
Its Hessian is given by
\begin{equation*}
  \nabla^2 l^c_i=\left(
    \begin{array}{ccc}
      \frac{\partial^2 l^c_i}{\partial \boldsymbol{\alpha} \partial \boldsymbol{\alpha}^\top} & 0 & 0 \\
      0 & \frac{\partial^2 l^c_i}{\partial \boldsymbol{\beta} \partial \boldsymbol{\beta}^\top} & \frac{\partial^2 l^c_i}{\partial \boldsymbol{\beta} \partial \boldsymbol{\lambda}^\top} \\
      0 & \left[ {\frac{\partial^2 l^c_i}{\partial \boldsymbol{\beta} \partial \boldsymbol{\lambda}^\top}} \right]^\top & \text{diag}(\frac{\partial^2 l^c_i}{\partial \lambda_k^2 })
    \end{array}\right),
\end{equation*}
where 
\begin{align*}
 \frac{\partial^2 l^c_i}{\partial \boldsymbol{\alpha} \partial \boldsymbol{\alpha}^\top}
=& \vecA_i \vecA_i^\top \Big\{-(1+M_i)\frac{e^{\bfalpha^\top \vecA_i}}{(1+e^{\bfalpha^\top \vecA_i})^2}\Big\}
 = - \vecA_i \vecA_i^\top (1+M_i)p_i(1-p_i), \\
  \frac{\partial^2 l^c_i}{\partial \boldsymbol{\beta} \partial \boldsymbol{\beta}^\top}
 =& \vecZ_i \vecZ_i^\top \bigg\{-\Big(\etaA_i \sum_{k:t_k \le X_i}\lambda_k
+M_i \sum_{k=1}^{\kappa_i} \lambda_k \Big)e^{\bfbeta^\top \vecZ_i}\bigg\}, \\
  \frac{\partial^2 l^c_i}{\partial \bfbeta \partial \lambda_k}
 =& \vecZ_i \bigg\{-\Big(\etaA_i I\{t_k \le X_i\}
+M_i I\{t_k \le \kappa_i\} \Big)e^{\bfbeta^\top \vecZ_i}\bigg\}, \\
 \frac{\partial^2 l^c_i}{\partial \lambda_k^2 }
=&  -\Big(\etaA_i \delta^1_i  I\{X_i=t_k\}+M_i
I\{\kappa_i=k\}\Big)\frac{1}{\lambda_k^2}, \\
  \frac{\partial^2 l^c_i}{\partial \boldsymbol{\alpha} \partial \boldsymbol{\beta}^\top}
 =& \frac{\partial^2 l^c_i}{\partial \boldsymbol{\alpha} \partial \boldsymbol{\lambda}^\top}
 = \frac{\partial^2 l^c_i}{\partial \lambda_k \partial \lambda_h }=0, \ \ \ \ k\neq h.
\end{align*}

\subsection{Conditional Expectations}

By the conditional expectations \eqref{EM:condE1} - \eqref{EM:condE3},
    we are able to calculate the `first order' conditional expectations,
    $\E[\nabla l^c_i|\Ocal]$ and $\E [\nabla^2 l^c_i|\Ocal]$:
\begin{align*}
 \E&  \left [ \frac{\partial l^c_i}{\partial \bfalpha}
\right] = \vecA_i\Big\{\E (\etaA_i)-p_i+\E (M_i)(1-p_i)\Big\}, \\
 \E&  \left [ \frac{\partial l^c_i}{\partial \bfbeta}
\right] = \vecZ_i \bigg[\E (\etaA_i) \Big\{\delta^1_i+\log S_i(X_i)\Big\}+\E (M_i) \Big\{1+\sum_{k:t_k<Q_j}\P(\tilde{T}_{ij}=t_k) \log S_i(t_k)\Big\}\bigg], \\
 \E&  \left [ \frac{\partial l^c_i}{\partial \lambda_k}
\right] =\E (\etaA_i)\Big\{\frac{\delta^1_i I\{t_k = X_i\}}{\lambda_k}-I\{t_k \le X_i\}e^{\bfbeta^\top \vecZ_i}\Big\}
+\E (M_i)\Big\{\frac{\P(\tilde{T}_{ij}=t_k)}{\lambda_k}-\P(\tilde{T}_{ij} \ge t_k)e^{\bfbeta^\top \vecZ_i}\Big\}.
\end{align*}
\begin{align*}
 \E& \left [ \frac{\partial^2 l^c_i}{\partial \boldsymbol{\alpha} \partial \boldsymbol{\alpha}^\top}
\right] = - \vecA_i \vecA_i^\top (1+\E (M_i))p_i(1-p_i), \\
 \E& \left [ \frac{\partial^2 l^c_i}{\partial \boldsymbol{\beta} \partial \boldsymbol{\beta}^\top}
\right] = \vecZ_i \vecZ_i^\top \Big\{ \E (\etaA_i) \log S_i(X_i)
+\E (M_i) \sum_{k:t_k<Q_i}\P(\tilde{T}_{ij}=t_k) \log S_i(t_k)\Big\}, \\
 \E& \left [ \frac{\partial^2 l^c_i}{\partial \bfbeta \partial \lambda_k}
\right] = -\vecZ_i \Big\{\E (\etaA_i) I\{t_k \le X_i\}
+\E (M_i) \P(t_k \le \kappa_i) \Big\}e^{\bfbeta^\top \vecZ_i}, \\
 \E& \left [ \frac{\partial^2 l^c_i}{\partial \lambda_k^2 }
\right] =  -\Big\{\E (\etaA_i) \delta^1_i  I\{\tilde{T}_{ij}=t_k\}+\E (M_i)
\P(\tilde{T}_{ij}=t_k)\Big\}\frac{1}{\lambda_k^2}.
\end{align*}
To calculate `second order' expectation $\E[\nabla l^c_i{\nabla l^c_i}^\top|\Ocal]$, we first compute the conditional variances:
\begin{align*}
 \Var&[\etaA_i|\mathcal{O}] =  \delta^c_i\frac{p_i(1-p_i)S_i(X_i)}{\big\{1-p_i+p_iS_i(X_i)\big\}^2}, \\
 \Var&[M_i|\mathcal{O}] = \frac{p_i\Big[1-S_i(Q_i)\big\}}{\big\{1-p_i+p_iS_i(Q_i)\big\}^2}.
\end{align*}
Then,
\begin{align*}
 \E  \left [ \frac{\partial l^c_i}{\partial \bfalpha} { \frac{\partial l^c_i}{\partial \bfalpha}}^\top 
\right] = &\E  \left [ \frac{\partial l^c_i}{\partial \bfalpha} \right]
\E \left [ \frac{\partial l^c_i}{\partial \bfalpha} \right]^\top
+\vecA_i \vecA_i^\top\big\{(1-p_i)^2 \Var (M_i)+ \Var (\etaA_i)\big\},\\
\notag \E \left [ \frac{\partial l^c_i}{\partial \bfalpha} { \frac{\partial l^c_i}{\partial \bfbeta}}^\top
\right] =&\E  \left [ \frac{\partial l^c_i}{\partial \bfalpha}\right]
\E  \left [ \frac{\partial l^c_i}{\partial \bfbeta} \right]^\top
+\vecA_i \vecZ_i^\top \bigg[ \Var (\etaA_i)\big\{\delta^1_i+\log S_i(X_i)\big\}\\
&+ \Var (M_i)(1-p_i)\Big\{1+\sum_{k:t_k<Q_i}\P(\tilde{T}_{ij}=t_k)\log S_i(t_k)\Big\}\bigg],\\
\notag \E \left [ \frac{\partial l^c_i}{\partial \bfbeta} {\frac{\partial l^c_i}{\partial \bfbeta}}^\top\right]
 =&\E  \left [ \frac{\partial l^c_i}{\partial \bfbeta} \right]
\E  \left [ \frac{\partial l^c_i}{\partial \bfbeta} \right]^\top
+\vecZ_i \vecZ_i^\top\bigg[ \Var (\etaA_i)\Big\{\delta^1_i+\log S_i(X_i)\Big\}^2\\
\notag &+ \Var (M_i)\Big\{1+\sum_{k:t_k<Q_i}\P(\tilde{T}_{ij}=t_k)\log S_i(t_k)\Big\}^2\\
 &+\E (M_i)\Big\{\sum_{k:t_k<Q_i}\P(\tilde{T}_{ij}=t_k)\log S_i(t_k)^2-\big(\sum_{k:t_k<Q_i}\P(\tilde{T}_{ij}=t_k)\log S_i(t_k)\big)^2\Big\}\bigg],\\
\notag \E \left [ \frac{\partial l^c_i}{\partial \bfalpha}  \frac{\partial l^c_i}{\partial \lambda_k}
\right] =&\E  \left [ \frac{\partial l^c_i}{\partial \bfalpha}\right]
\E  \left [\frac{\partial l^c_i}{\partial \lambda_k}\right]  + \vecA_i
\bigg[\Var (\etaA_i) \Big\{\frac{\delta^1_i I\{t_k = X_i\}}{\lambda_k}-I\{t_k \le X_i\}e^{\bfbeta^\top \vecZ_i}\Big\} \\
 &+\Var (M_i)(1-p_i)\Big\{\frac{\P(\tilde{T}_{ij}=t_k)}{\lambda_k}-\P(\tilde{T}_{ij} \ge t_k)e^{\bfbeta^\top \vecZ_i}\Big\}\bigg],\\
\notag \E \left[ \frac{\partial l^c_i}{\partial \bfbeta} \frac{\partial l^c_i}{\partial \lambda_k}
\right] =&\E \left[ \frac{\partial l^c_i}{\partial \bfbeta}\right]
\E   \left[\frac{\partial l^c_i}{\partial \lambda_k}\right] \\
\notag &+\vecZ_i \bigg[ \Var (\etaA_i)\big\{\delta^1_i+\log S_i(X_i)\big\}\Big\{\frac{\delta^1_i I\{t_k  = X_i\}}{\lambda_k}-I\{t_k \le X_i\}e^{\bfbeta^\top \vecZ_i}\Big\} \\
\notag &+ \Var (M_i)\Big\{\frac{\P(\tilde{T}_{ij}=t_k)}{\lambda_k}-\P(\tilde{T}_{ij} \ge t_k)e^{\bfbeta^\top \vecZ_i}\Big\}
\Big\{1+\sum_{h:t_h<Q_i}\P(\tilde{T}_{ij}=t_h)\log S_i(t_h)\Big\}\\
\notag &- \E (M_i)\Big\{\sum_{h:t_h<Q_i}\P(\tilde{T}_{ij}=t_h)\log S_i(t_h)\frac{\P(\tilde{T}_{ij}=t_k)}{\lambda_k}-\frac{\P(\tilde{T}_{ij}=t_k)\log S_i(t_k)}{\lambda_k}\\
\notag &
-\P\{\tilde{T}_{ij} \ge t_k\}e^{\bfbeta^\top \vecZ_i} \sum_{h:t_h<Q_i}\P(\tilde{T}_{ij}=t_h)\log S_i(t_h)\\
 &+e^{\bfbeta^\top \vecZ_i}\sum_{h =k}^{t_h<Q_i}\P(\tilde{T}_{ij}=t_h)\log S_i(t_h)\Big\}\bigg],
\end{align*}
\begin{align*}
 \notag &\E \left [ \frac{\partial l^c_i}{\partial \lambda_k} \frac{\partial l^c_i}{\partial \lambda_h} 
\right] =\E\etaA_i \left\{ -\frac{\delta^1_i I\{X_i =t_{k\vee h}\}}{\lambda_{k\vee h}}e^{\bfbeta^\top\vecZ_i}+
I\{X_i\ge t_{k\vee h}\}e^{2\bfbeta^\top\vecZ_i}\right\}\\
\notag &+\E (\etaA_i) \E (M_i) \left\{ \frac{\delta^1_i I\{\tilde{T}_{ij}=t_k\}}{\lambda_k}-
I\{X_i\ge t_k\}e^{\bfbeta^\top\vecZ_i}\right\}
\left\{ \frac{\P(\tilde{T}_{ij}=t_h)}{\lambda_h}-\P(\tilde{T}_{ij} \ge t_h)e^{\bfbeta^\top\vecZ_i} \right\}\\
\notag &+\E (\etaA_i) \E (M_i) \left\{ \frac{\delta^1_i I\{X_i =t_h\}}{\lambda_h}-
I\{X_i\ge t_h\}e^{\bfbeta^\top\vecZ_i}\right\}
\left\{ \frac{\P(\tilde{T}_{ij}=t_k)}{\lambda_k}-\P(\tilde{T}_{ij} \ge t_k)e^{\bfbeta^\top\vecZ_i} \right\} \\
\notag &+ \E [M_i^2-M_i]  \left\{ \frac{\P(\tilde{T}_{ij}=t_k)}{\lambda_k}-\P(\tilde{T}_{ij} \ge t_k)e^{\bfbeta^\top\vecZ_i} \right\}
 \left\{ \frac{\P(\tilde{T}_{ij}=t_h)}{\lambda_h}-\P(\tilde{T}_{ij} \ge t_h)e^{\bfbeta^\top\vecZ_i} \right\}\\
\notag &+\E (M_i) \left\{ -\frac{\P(\tilde{T}_{ij} =t_{k\vee h})}{\lambda_{k \vee h}}e^{\bfbeta^\top\vecZ_i} +\P(\kappa_i\ge t_{k \vee h})e^{2\bfbeta^\top\vecZ_i} \right\},\\
\notag &\E \left [ \frac{\partial l^c_i}{\partial \lambda_k} \frac{\partial l^c_i}{\partial \lambda_k} 
\right] =\E\etaA_i \left\{ \frac{\delta^1_i I\{\tilde{T}_{ij}=t_k\}}{\lambda_k}-
I\{X_i\ge t_k\}e^{\bfbeta^\top\vecZ_i}\right\}^2\\
\notag &+\E (\etaA_i) \E (M_i) \left\{ \frac{\delta^1_i I\{\tilde{T}_{ij}=t_k\}}{\lambda_k}-
I\{X_i\ge t_k\}e^{\bfbeta^\top\vecZ_i}\right\}
\left\{ \frac{\P(\tilde{T}_{ij}=t_k)}{\lambda_k}-\P(\tilde{T}_{ij} \ge t_k)e^{\bfbeta^\top\vecZ_i} \right\}\\
\notag &+\E (\etaA_i) \E (M_i) \left\{ \frac{\delta^1_i I\{\tilde{T}_{ij}=t_k\}}{\lambda_k}-
I\{X_i\ge t_k\}e^{\bfbeta^\top\vecZ_i}\right\}
\left\{ \frac{\P(\tilde{T}_{ij}=t_k)}{\lambda_k}-\P(\tilde{T}_{ij} \ge t_k)e^{\bfbeta^\top\vecZ_i} \right\} \\
\notag &+ \E [M_i^2-M_i]  \left\{ \frac{\P(\tilde{T}_{ij}=t_k)}{\lambda_k}-\P(\tilde{T}_{ij} \ge t_k)e^{\bfbeta^\top\vecZ_i} \right\}
 \left\{ \frac{\P(\tilde{T}_{ij}=t_k)}{\lambda_k}-\P(\tilde{T}_{ij} \ge t_k)e^{\bfbeta^\top\vecZ_i} \right\}\\
 &+\E (M_i) \left\{ \frac{\P(\tilde{T}_{ij}=t_k)}{\lambda^2_k}- 2\frac{\P(\tilde{T}_{ij}=t_k)}{\lambda_k}e^{\bfbeta^\top\vecZ_i}
+\P(\tilde{T}_{ij} \ge t_k)e^{2\bfbeta^\top\vecZ_i} \right\}.
\end{align*}
\end{appendix}


\newpage

\begin{figure}
\caption{Study entry times for all individuals  in the SAB data (left), and left truncated Kaplan-Meier curves (95\% confidence intervals) for the SAB events (right).}
\label{fig:Qhist}
\begin{center}
\includegraphics[scale=.6]{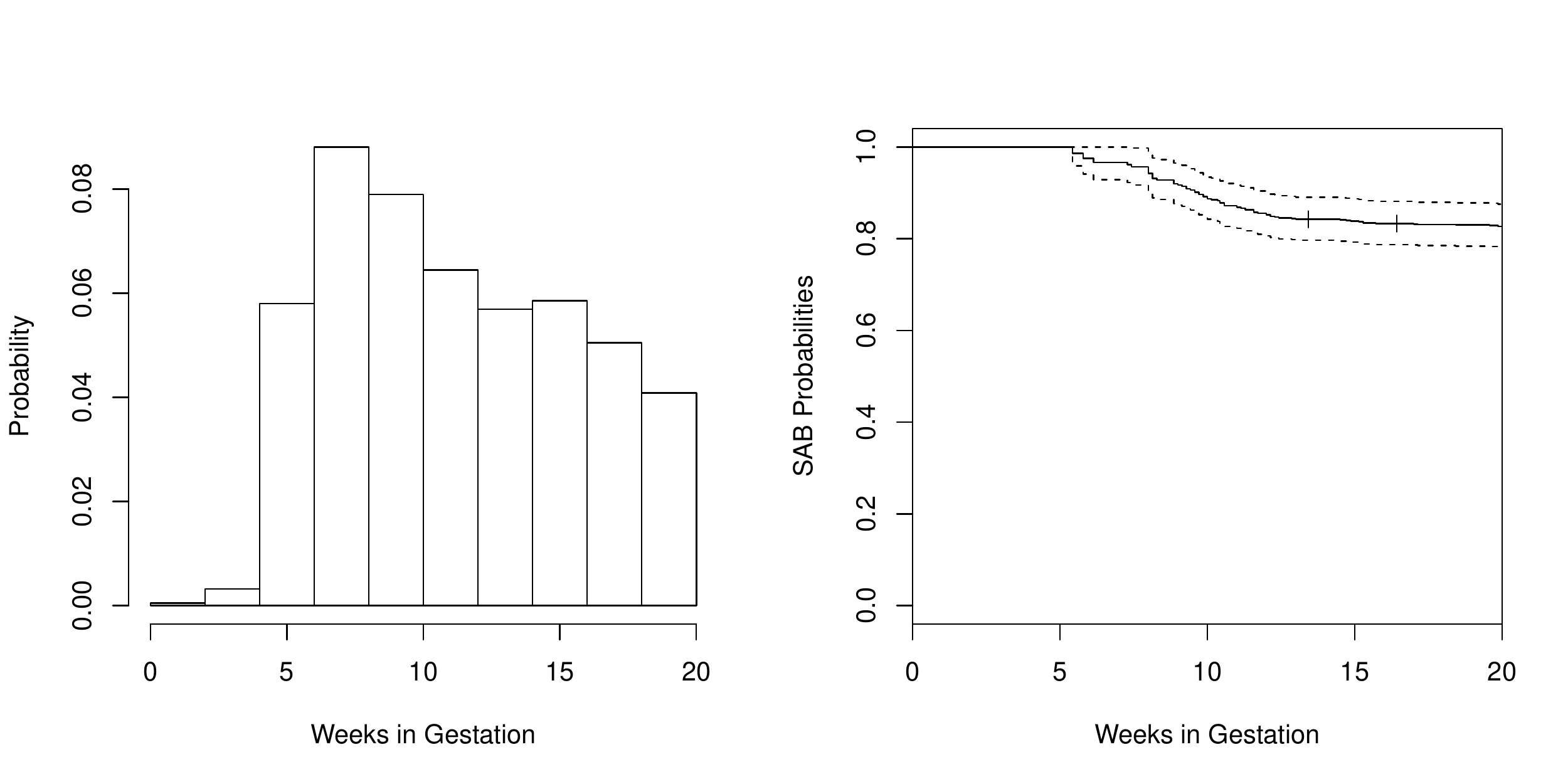}
\end{center}
\end{figure}

\newpage
\begin{table}[ht]
\caption{Simulation results using the EM algorithm for NPMLE. } \label{table:simulation}
{\footnotesize
  \begin{center}  \begin{tabular}{ccccccccccc}
\hline \hline
~ & ~ & \multicolumn{4}{c}{$n=200$} & \multicolumn{1}{c}{~} & \multicolumn{4}{c}{$n=1000$} \\ \cline{3-6} \cline{8-11}
        ~               & True Value                       & Estimate &  Sample SD & SE & \multicolumn{1}{c}{Coverage} & ~ & Estimate &  Sample SD & SE & Coverage\\
\hline \hline
\multicolumn{11}{c}{10\% Truncation, 0\% Censoring} \\ \hline
$\alpha_{0}$          & 1.00          &   1.01&  0.79&  0.75&  94.0 \%&&  0.98&  0.35&  0.33&  93.6 \% \\
$\alpha_{1}$          & -0.63        &   -0.64&  0.20&  0.19&  94.8 \%&&  -0.63&  0.09&  0.08&  94.3 \% \\
$\alpha_{2}$          & 1.00          &   1.00&  0.36&  0.37&  95.6 \%&&  1.01&  0.16&  0.16&  94.8  \%\\
$\beta_{1}$          & -0.20          &    -0.23&  0.20&  0.17&  92.2 \%&&  -0.20&  0.07&  0.07&  95.6  \%\\
$\beta_{2}$            & 0.30          &   0.33&  0.34&  0.32&  94.2 \%&&  0.29&  0.14&  0.13&  93.4  \%\\
        \hline \hline
\multicolumn{11}{c}{20\% Truncation, 0\% Censoring} \\ \hline
$\alpha_{0}$          & 1.00          &  0.97&  0.80&  0.79&  94.8 \% &&  0.99&  0.34&  0.35&  96.0  \%\\
$\alpha_{1}$          & -0.63        &   -0.64&  0.21&  0.20&  95.4 \% &&  -0.63&  0.09&  0.09&  96.2  \%\\
$\alpha_{2}$          & 1.00          &    0.98&  0.40&  0.39&  95.4 \% && 0.99&  0.17&  0.17&  95.2  \%\\
$\beta_{1}$          & -0.20          &    -0.20&  0.20&  0.18&  94.6  \%&&  -0.20&  0.07&  0.07&  95.2  \%\\
$\beta_{2}$            & 0.30          &   0.31&  0.37&  0.34&  94.6 \% &&  0.30&  0.14&  0.14&  94.2  \%\\
        \hline \hline
\multicolumn{11}{c}{10\% Truncation, 20\% Censoring} \\ \hline
$\alpha_{0}$          & 1.00          &   1.18&  0.97&  0.99&  96.6 \% &&  1.02&  0.41&  0.42&  95.8   \%\\
$\alpha_{1}$          & -0.63        &   -0.69&  0.25&  0.26&  96.2 \% &&  -0.64&  0.11&  0.11&  96.2 \% \\
$\alpha_{2}$          & 1.00          &   1.01&  0.50&  0.49&  95.4 \% &&  1.00&  0.20&  0.21&  96.0 \% \\
$\beta_{1}$          & -0.20          &    -0.21&  0.30&  0.26&  91.6  \%&&  -0.21&  0.11&  0.11&  94.4 \% \\
$\beta_{2}$            & 0.30          &   0.31&  0.53&  0.49&  93.4 \% &&  0.30&  0.21&  0.20&  93.8 \% \\
        \hline \hline
\multicolumn{11}{c}{20\% Truncation, 20\% Censoring} \\ \hline
$\alpha_{0}$          & 1.00          &   1.05&  1.00&  0.96&  95.8 \% &&  0.98&  0.37&  0.41&  97.0 \% \\
$\alpha_{1}$          & -0.63        &   -0.66&  0.27&  0.25&  96.6 \% &&  -0.63&  0.10&  0.11&  96.6 \% \\
$\alpha_{2}$          & 1.00          &   1.05&  0.49&  0.47&  95.6  \%&&  1.01&  0.20&  0.20&  94.6 \% \\
$\beta_{1}$          & -0.20          &   -0.19&  0.30&  0.26&  90.4 \% &&  -0.20&  0.11&  0.11&  95.0 \% \\
$\beta_{2}$            & 0.30          &   0.33&  0.54&  0.48&  92.2 \% & & 0.31&  0.21&  0.20&  94.8 \% \\
        \hline
    \end{tabular} \end{center}
}
\end{table}

\newpage
{\footnotesize
\begin{table}[htb]
\caption{Cure rate model versus naive model fits for SAB data} \label{table:SAB}
  \begin{center}
  \begin{tabular}{lrcrc}
        \hline\hline
  &      \multicolumn{2}{c}{Cure model} &  \multicolumn{2}{c}{Separate models} \\
        ~               &    Estimate (SE) & P-value &    Estimate (SE) & P-value \\
\hline 
Logistic &  ~ & ~ &&\\
Intercept  &  -0.74 (0.54)  &  0.17    &    -2.25 (0.49)  &  $<$0.01\\
Healthy  &  -0.54 (0.49)  &  0.27    &    -0.92 (0.45)  &  0.04\\
Diseased Control  &  0.18 (0.31)  &  0.56    &    0.01 (0.28)  &  0.98\\
BMI  &  -0.37 (0.18)  &  0.04    &    -0.11 (0.16)  &  0.51\\
Gravidity$>$1  &  0.01 (0.3)  &  0.97    &    0.2 (0.27)  &  0.46\\
Smoking  &  0.41 (0.37)  &  0.27    &    0.65 (0.34)  &  0.06\\
Alcohol  &  -0.1 (0.29)  &  0.73    &    -0.24 (0.26)  &  0.35\\
        \hline 
Cox PH & ~ & ~ &&\\
Healthy Control  &  -0.36 (0.38)  &  0.34    &    -0.41 (0.5)  &  0.42\\
Diseased Control  &  -0.34 (0.25)  &  0.17    &    -0.29 (0.32)  &  0.36\\
BMI  &  -0.35 (0.09)  &  $<$0.01&    -0.84 (0.22)  &  $<$0.01\\
Gravidity$>$1  &  -0.52 (0.23)  &  0.02    &    -0.38 (0.28)  &  0.18\\
Smoking  &  -1.01 (0.33)  &  $<$0.01&    -0.65 (0.36)  &  0.08\\
Alcohol  &  0.66 (0.26)  &  0.01    &    0.87 (0.29)  &  $<$0.01\\
        \hline
    \end{tabular}
    \end{center}
\end{table}
}

\newpage
\begin{figure}
\caption{Left truncated Kaplan-Meier curves for SAB events according to
BMI (top) or alcohol (bottom), among the full data set (left) and without the
observed cured individuals (right); the $p$-values are from the cure model (Table \ref{table:SAB}).}
\label{fig:KMbycov}
\begin{center}
\includegraphics[scale=.5]{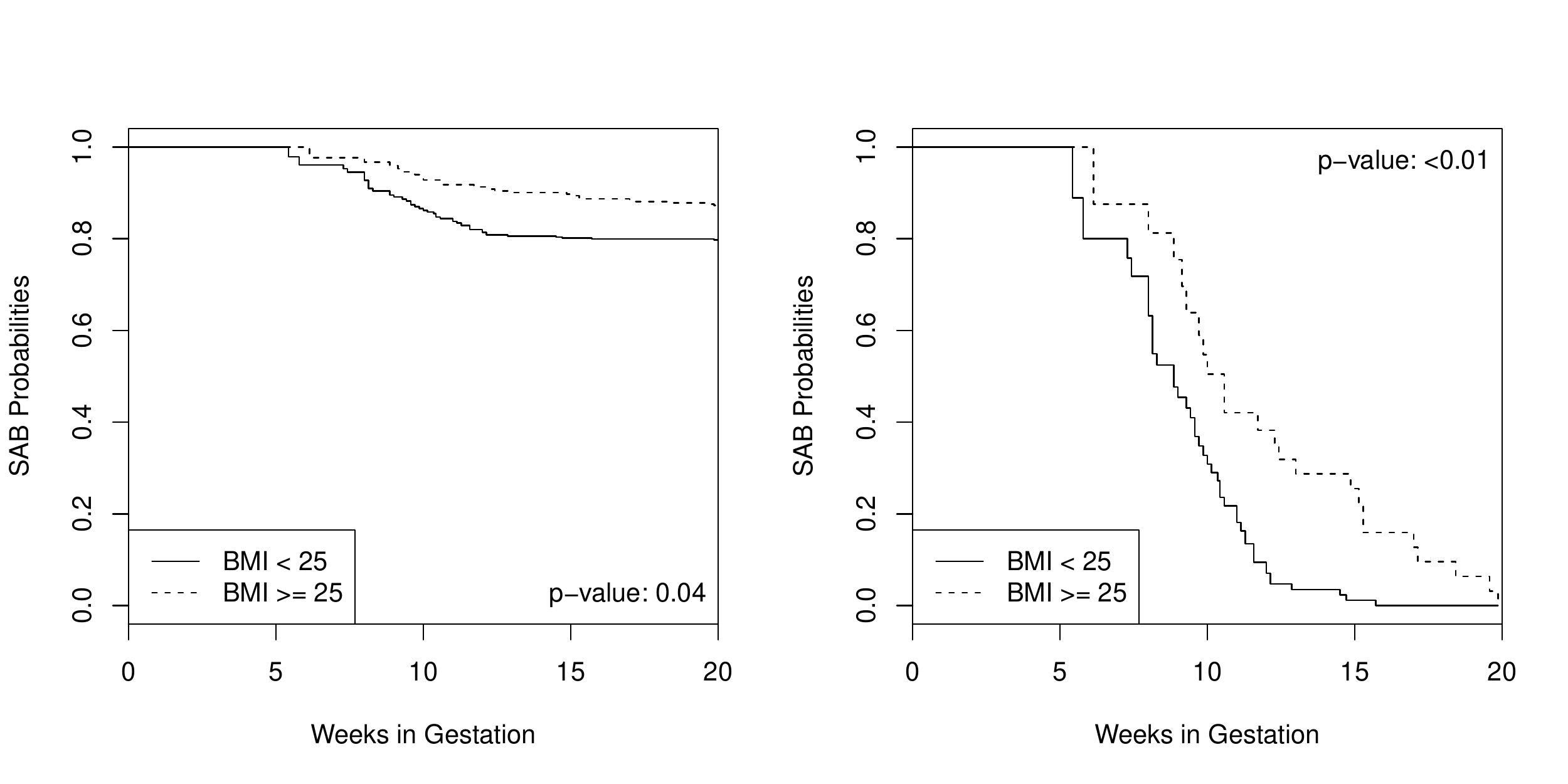}\\
\includegraphics[scale=.5]{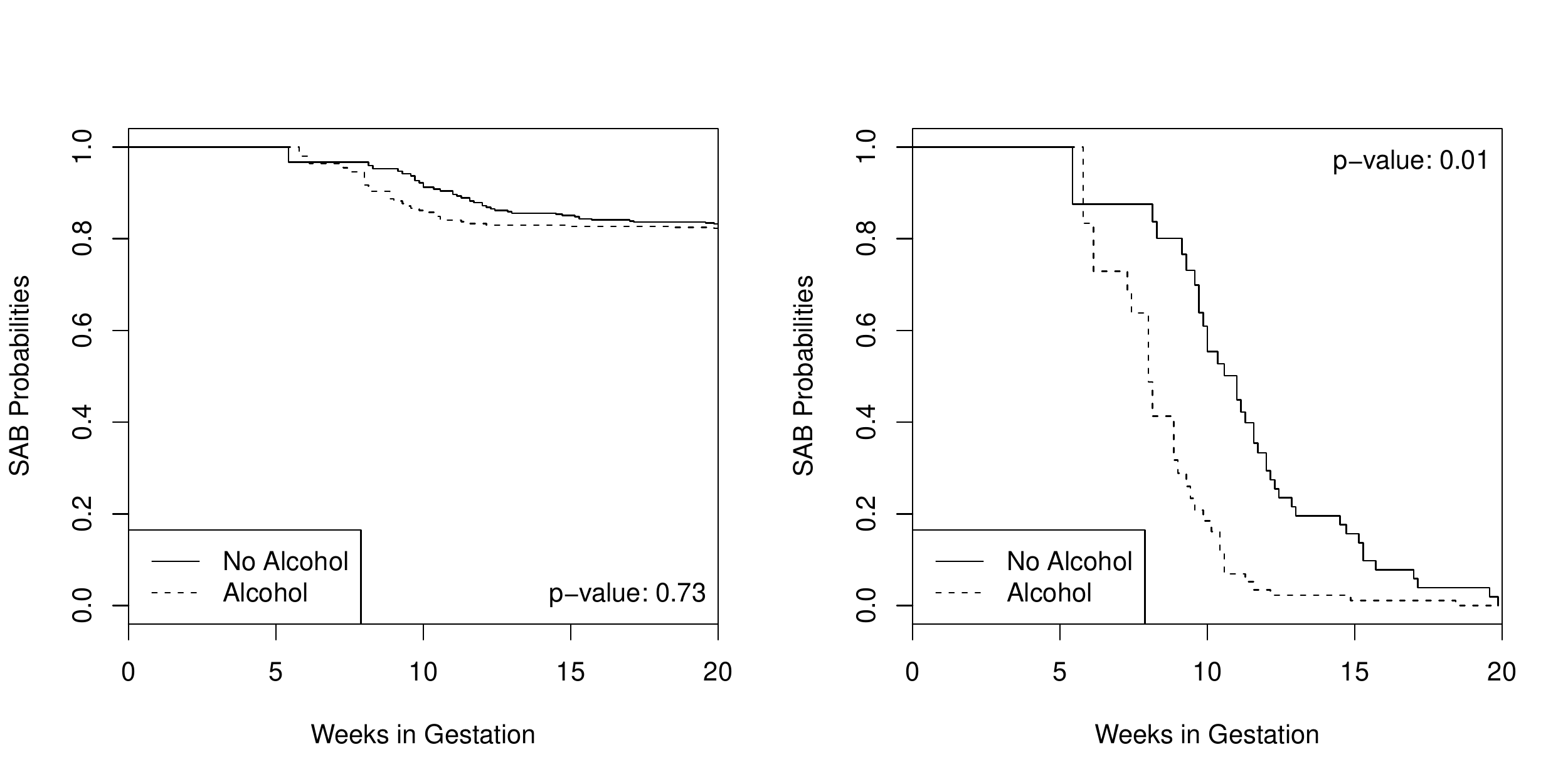}
\end{center}
\end{figure}

\end{document}